\documentclass[journal]{IEEEtran}

\usepackage{graphicx}

\usepackage[T1]{fontenc}
\usepackage{color}
\usepackage{bm}
\usepackage{mathrsfs}
\usepackage{varioref}
\usepackage{textcomp}
\usepackage{amsthm}
\usepackage{amsmath}
\usepackage{amssymb}
\usepackage{graphicx}
\usepackage{setspace}
\usepackage{esint}
\usepackage{enumerate}
\usepackage{algorithm}
\usepackage{algpseudocode}
\usepackage{pifont}

\usepackage{enumitem}
\usepackage{enumerate}
\usepackage{enumitem}
\usepackage{bm}
\usepackage{graphicx,subfigure}
\usepackage{cite}
\usepackage{amssymb}
\usepackage{psfrag}
\usepackage{latexsym, amsmath, subfigure, color, amsfonts, amssymb,graphicx}
\makeatletter

\newcommand{\Rmnum}[1]{\expandafter\@slowromancap\romannumeral #1@}
\makeatother
 
\newtheorem{theorem}{Theorem}
\newtheorem{definition}{Definition}

\newtheorem{remark}{Remark}

\providecommand{\propositionname}{Proposition}

\ifCLASSINFOpdf
\else
\fi

\usepackage[T1]{fontenc}
\usepackage{etoolbox}

\makeatletter
\patchcmd{\maketitle}{\@fnsymbol}{\@alph}{}{}  
\makeatother

\title{Decentralized Caching and Coded Delivery with Distinct Cache Capacities}
\author{
  Mohammad Mohammadi Amiri,~\IEEEmembership{Student Member,~IEEE}, Qianqian Yang,~\IEEEmembership{Student Member,~IEEE},\thanks{The authors are with Imperial College London, London SW7 2AZ, U.K. (e-mail: m.mohammadi-amiri15@imperial.ac.uk; q.yang14@imperial.ac.uk; d.gunduz@imperial.ac.uk).} and~Deniz~G\"und\"uz,~\IEEEmembership{Senior Member,~IEEE}
}
\date{}

\begin{document}

\maketitle
\vspace{-.6in}
\begin{abstract}
Decentralized proactive caching and coded delivery is studied in a content delivery network, where each user is equipped with a cache memory, not necessarily of equal capacity. Cache memories are filled in advance during the off-peak traffic period in a \textit{decentralized} manner, i.e., without the knowledge of the number of active users, their identities, or their particular demands. User demands are revealed during the peak traffic period, and are served simultaneously through an error-free shared link. The goal is to find the minimum delivery rate during the peak traffic period that is sufficient to satisfy all possible demand combinations. A group-based decentralized caching and coded delivery scheme is proposed, and it is shown to improve upon the state-of-the-art in terms of the minimum required delivery rate when there are more users in the system than files. Numerical results indicate that the improvement is more significant as the cache capacities of the users become more skewed. A new lower bound on the delivery rate is also presented, which provides a tighter bound than the classical cut-set bound.  
\end{abstract}


\section{Introduction}\label{Intro}
The\makeatletter{\renewcommand*{\@makefnmark}{}
\footnotetext{This work received support from the European Union's H2020 Research and Innovation Programme through project TACTILENet: Towards Agile, effiCient, auTonomous and massIvely LargE Network of things (agreement 690893), and from the European Research Council (ERC) through Starting Grant BEACON (agreement 677854).}\makeatother} ever-increasing mobile data traffic is imposing a great challenge on the current network architectures. The growing demand has been typically addressed by increasing the achievable data rates; however, moving content to the network edge has recently emerged as a promising alternative solution as it reduces both the bandwidth requirements and the delay. In this paper, we consider an extreme form of edge caching, in which contents are stored directly at user terminals in a proactive manner. Proactive caching of popular contents, e.g., trending Youtube videos, episodes of popular TV series, during off-peak traffic periods also helps flattening the high temporal variability of traffic \cite{DowdyCaching,AlmerothCacing}.     

In this proactive caching model \cite{MaddahAliCentralized}, the \textit{placement phase} takes place during off-peak traffic hours when the resources are abundant, without the knowledge of particular user demands. When the user demands are revealed, the \textit{delivery phase} is performed, in which a common message is transmitted from the server to all the users over the shared communication channel. Each user decodes its requested file by combining the bits received in the delivery phase with the contents stored in its local cache. Cache capacities are typically much lower than the size of the whole database, and a key challenge is to decide how to fill the cache memories without the knowledge of the user demands in order to minimize the \textit{delivery rate}, which guarantees that all the user demands are satisfied, independent of the specific demand combination across the users. Maddah-Ali and Niesen showed in \cite{MaddahAliCentralized} that by storing and transmitting coded contents, and designing the placement and delivery phases jointly, it is possible to significantly reduce the delivery rate compared to uncoded caching. 

A \textit{centralized} caching scenario is studied in \cite{MaddahAliCentralized}, in which the number and the identities of the users are known in advance by the server. This allows coordination of the cache contents across the users during the placement and delivery phases, such that by carefully placing pieces of contents in user caches a maximum number of multicasting opportunities are created to be exploited during the delivery phase. Many more recent works study centralized coded caching, and the required delivery rate has been further reduced \cite{ZhiChenXOR,MohammadDenizISITA,MohammadDenizTCom,TianCentralizedCachingNew,MohammadQianDenizITW, Gomez_16}.

In practice, however, the number or identities of active users that will participate in the \textit{delivery phase} might not be known in advance during the \textit{placement phase}. In such a scenario, called \textit{decentralized caching}, coordination across users is not possible during the \textit{placement phase}. For this scenario Maddah-Ali and Niesen proposed caching an equal number of random bits of each content at each user, and showed that one can still exploit multicasting opportunities in the delivery phase, albeit limited compared to the centralized counterpart \cite{MaddahAliDecentralized}. \textit{Decentralized caching} has been studied in various other settings, for example, with files with different popularities \cite{NiesenNonuniform, JiArXivNonuniform}, and distinct lengths \cite{ZhangDistinctFileSizes}, for online coded caching \cite{PedarsaniOnlineCaching}, coded caching of files with lossy reconstruction \cite{QianDenizLossyJournal}, as well as delivering contents from multiple servers over an interference channel \cite{Joan_ICC17}, and in the presence of a fading delivery channel \cite{HuangFadingChannelcodedcaching}. 

Most of the literature on coded caching assume identical cache capacities across users. However, in practice users access content through diverse devices, typically with very different storage capacities. Centralized caching with distinct cache capacities is studied in \cite{QianDenizLossyJournal} and \cite{Ibrahim_WCNC17}. Recently, in \cite{WangHeterogenous}, decentralized caching is studied for heterogeneous cache capacities; and by extending the scheme proposed in \cite{MaddahAliDecentralized} to this scenario, authors have shown that significant gains can still be obtained compared to uncoded caching despite the asymmetry across users. In this paper, we propose a novel decentralized caching and delivery algorithm for users with distinct cache capacities. We show that the proposed scheme requires a smaller delivery rate than the one achieved in \cite{WangHeterogenous} when there are more users in the system than the number of files in the library, while the same performance is achieved otherwise. This scenario is relevant when a few popular video files or software updates are downloaded by many users within a short time period. Simulation results illustrate that the more distinct the cache capacities of the users are, which is more likely to happen in practice, the higher the improvement (with respect to \cite{WangHeterogenous}). We also derive an information-theoretic lower bound on the delivery rate building upon the lower bound derived in \cite{SenguptaCaching} for homogeneous cache capacities. This provides a lower bound on the delivery rate that is tighter than the classical cut-set bound.      

The rest of this paper is organized as follows. The system model is introduced in Section \ref{SystemModel}. In Section \ref{s:Results}, we present the proposed caching scheme as well as a lower bound on the delivery rate- cache capacity trade-off. The performance of the proposed coded caching scheme is compared with the state-of-the-art result analytically, and some numerical results are presented in Section \ref{s:Comparison}. We  conclude the paper in Section \ref{Conc}. The detailed proofs are given in the Appendices.

\textit{Notations:} The set of integers $\left\{ i, ..., j \right\}$, where $i \le j$, is denoted by $\left[ i:j \right]$. We denote the sequence $Y_{i}, Y_{i+1}, \dots, Y_{j-1}, Y_{j}$ shortly by $Y_{[i:j]}$. $\binom{j}{i}$ represents the binomial coefficient. For two sets $\mathcal{Q}$ and $\mathcal{P}$, $\mathcal{Q} \backslash \mathcal{P}$ is the set of elements in $\mathcal{Q}$ that do not belong to $\mathcal{P}$. Notation $\left| \cdot \right|$ represents cardinality of a set, or the length of a file. Notation $\oplus$ refers to bitwise XOR operation, while $\bar \oplus$ represents bitwise XOR operation where the arguments are first zero-padded to have the same length as the longest argument. Finally, $\left\lfloor x \right\rfloor $ denotes the floor function; and ${\left( x \right)^ + } \buildrel \Delta \over = \max \left\{ x,0 \right\}$.  

\section{System Model}\label{SystemModel}

A server with a content library of $N$ independent files $W_{[1:N]}$ is considered. All the files in the library are assumed to be of length $F$ bits, and each of them is chosen uniformly randomly over the set $\left[1:2^F \right]$. There are $K$ active users, $U_{[1:K]}$, where $U_k$ is equipped with a cache memory of capacity $M_{k}F$ bits, with $M_k < N$, $\forall k$\footnote{If $M_k \ge N$, for $k \in \left[ 1:K \right]$, $U_k$ has enough memory to cache all the database; so, $U_k$ does not need to participate in the delivery phase.}. Data delivery is divided into two phases. User caches are filled during the \textit{placement phase}. Let $Z_k$ denote the contents of $U_k$'s cache at the end of the placement phase, which is a function of the database $W_{[1:N]}$ given by ${Z_k} = {\phi _k}\left( {W_{[1:N]}} \right)$, for $k \in \left[ {1:K} \right]$. Unlike in centralized caching \cite{MaddahAliCentralized}, cache contents of each user are independent of the number and identities of other users in the system. User requests are revealed after the placement phase, where $d_k \in \left[ 1:N \right]$ denotes the demand of $U_k$, for $k \in \left[ {1:K} \right]$. These requests are served simultaneously through an error-free shared link in the \textit{delivery phase}. The $RF$-bit message sent over the shared link by the server in response to the user demands $d_{[1:K]}$ is denoted by $X$, where $X \in [1:2^{RF}]$, and it is generated by the encoding function $\psi$, i.e., $X=\psi \left(W_{[1:N]}, d_{[1:K]}\right)$. $U_k$ reconstructs its requested file $W_{d_k}$ after receiving the common message $X$ in the delivery phase along with its cache contents $Z_k$. The reconstruction at $U_k$ for the demand combination $d_{[1:K]}$ is given by ${\hat W_{{d_k}}} = {\rho _k}\left(Z_k, X, d_{[1:K]}\right)$, $\forall k \in \left[1:K \right]$, where ${\rho _k}$ is the decoding function at user $U_k$. For a given content delivery network, the tuple $\left( {\phi _{[1:K]},\psi ,\rho_{[1:K]}} \right)$ constitute a caching and delivery code with delivery rate $R$. We are interested in the \textit{worst-case} delivery rate, that is the delivery rate that is sufficient to satisfy all demand combinations. Accordingly, the error probability is defined over all demand combinations as follows.

\begin{definition}
The error probability of a $\left( {\phi_{[1:K]},\psi ,\rho_{[1:K]}} \right)$ caching and delivery code described above is given by
\begin{equation}\label{ErrorProbabilityFunction}
P_e \triangleq \mathop {\max }\limits_{ d_{[1:K]} \in [1:N]^{K} } \Pr \left\{ {\mathop  \bigcup \limits_{k = 1}^K \left\{ {{{\hat W}_{d_k}} \ne {W_{{d_k}}}} \right\}} \right\}.
\end{equation}
\end{definition}

\begin{definition}
For a content delivery network with $N$ files and $K$ users, we say that a cache capacity-delivery rate tuple $\left( {M_{[1:K]} ,R} \right)$ is achievable if, for every $\varepsilon  > 0$, there exists a caching and delivery code $\left( {{\phi _{[1:K]}},\psi ,{\rho_{[1:K]}}} \right)$ with error probability ${P_e} < \varepsilon$, for $F$ large enough.    
\end{definition}

There is a trade-off between the achievable delivery rate $R$ and the cache capacities $M_{[1:K]}$, defined as
\begin{equation}\label{DeliveryRateCacheCapacityTradeoffDefinition}
{R^*}\left( M_{[1:K]}  \right) \buildrel \Delta \over = \min \left\{ \mbox{$R:\left( {M_{[1:K]} ,R} \right)$ is achievable} \right\}.
\end{equation}
In this paper, we present upper and lower bounds on this trade-off.

\section{The Group-Based Decentralized Caching (GBD) Scheme}\label{s:Results}

Here, we present the proposed \textit{group-based decentralized (GBD) caching scheme}, first for uniform cache capacities, and then extend it to the scenario with distinct cache capacities.

\subsection{Uniform Cache Capacities}\label{UniformScenario}
Here we assume that each user has the same cache capacity of $MF$ bits, i.e., $M_1 = \cdots = M_K = M$.  

\underline{\textbf{Placement phase:}} In the placement phase, as in  \cite{MaddahAliDecentralized}, each user caches a random subset of $MF/N$ bits of each file independently. Since there are $N$ files, each of length $F$ bits, this placement phase satisfies the memory constraint.

For any set $\mathcal{V} \subset \left[ {1:K} \right]$, $W_{i,\mathcal{V}}$ denotes the bits of file $W_i$ that have been \textit{exclusively} cached by the users in set $\mathcal{V}$, that is, $W_{i,\mathcal{V}} \subset Z_k$, $\forall k \in \mathcal{V}$, and $W_{i,\mathcal{V}} \cap  Z_k = \emptyset$, $\forall k \in \left[ 1:K \right] \backslash \mathcal{V}$. For any chosen bit of a file, the probability of having been cached by any particular user is $M/N$. Since the contents are cached independently by each user, a bit of each file is cached exclusively by the users in set $\mathcal V \subset \left[ 1:K \right]$ (and no other user) with probability $\left( M/N \right)^{\left| \mathcal V \right|} \left( {1 - M/N} \right)^{K - \left| \mathcal V \right|}$.

\underline{\textbf{Delivery phase:}} Without loss of generality, we order the users such that the first $K_1$ users, referred to as group ${\cal G}_1$, demand $W_1$, the next $K_2$ users, referred to as group ${\cal G}_2$, request $W_2$, and so on so forth. We define ${S_i} \buildrel \Delta \over = \sum\limits_{l = 1}^i {{K_l}}$, which denotes the total number of users in the first $i$ groups. Hence, the user demands are as follows:
\begin{equation}\label{DemandsGeneralUniformCase} d_k = i,\quad \mbox{for $i=1, ..., N,$ and $k=S_{i-1} + 1, ..., S_i$},
\end{equation} 
where we set $S_0 = 0$.

There are two alternative delivery procedures, called CODED DELIVERY and RANDOM DELIVERY, presented in Algorithm \ref{DeliveryDecentralizedUniform}. The server follows the one that requires a smaller delivery rate. We present the CODED DELIVERY procedure of Algorithm \ref{DeliveryDecentralizedUniform} in detail, while we refer the reader to \cite{MaddahAliDecentralized} for the RANDOM DELIVERY procedure, as we use the same procedure in \cite{MaddahAliDecentralized} for the latter.

\begin{algorithm}[!t]
\caption{Coded Delivery Phase for Uniform Cache Capacities Scenario}
\label{DeliveryDecentralizedUniform}
\begin{algorithmic}[1]
\Statex
\Procedure {CODED DELIVERY}{}
\State{\textbf{Part 1}: Delivering bits that are not in the cache of any user}
\For {$i = 1, \ldots, N$}
\State{send $ W_{{d_{S_{i-1} + 1}}, \emptyset}$}
\EndFor
\Statex
\State{\textbf{Part 2}: Delivering bits that are in the cache of only one user}
\State{send $ {\bigcup\limits_{i = 1}^N {\bigcup\limits_{k = {S_{i - 1}} + 1}^{{S_i} - 1} {\left( {{W_{i,\left\{ k \right\}}} \oplus {W_{i,\left\{ k + 1 \right\}}}} \right)} } } $}
\State{send $\bigcup\limits_{i = 1}^{N - 1} \bigcup\limits_{j = i + 1}^N \left( \bigcup\limits_{k = {S_{j - 1}} + 1}^{{S_j} - 1} {\left( {{W_{i,\left\{ k \right\}}} \oplus {W_{i,\left\{ {k + 1} \right\}}}} \right)} ,\qquad \qquad \qquad \qquad \qquad \qquad \qquad \right.$ $\left. \qquad \qquad \qquad \quad \quad \;\;  \bigcup\limits_{k = {S_{i - 1}} + 1}^{{S_i} - 1} {\left( {{W_{j,\left\{ k \right\}}} \oplus {W_{j,\left\{ {k + 1} \right\}}}} \right)} ,\right.$ $\left. \qquad \qquad \qquad \quad \qquad \left( {W_{i,\left\{ {{S_{j-1}+1}} \right\}}} \oplus W_{j,\left\{ {{S_{i-1}+1}} \right\}} \right) \Bigg) \right.$}
\Statex
\State{\textbf{Part 3}: Delivering bits that are in the cache of more than one user}
\For{$\mathcal V \subset \left[ 1:K \right]: 3 \le \left| \mathcal V \right| \le K$ }
\State send ${\bigoplus} _{v \in \mathcal V} W_{d_{v}, \mathcal V \backslash \left\{v\right\}}$
\EndFor
\EndProcedure
\Statex
\Procedure {RANDOM DELIVERY}{}\cite{MaddahAliDecentralized}
\For{$i = 1, \ldots, N$}
\State {send sufficient random linear combinations (for reliable decoding) of the bits of file $W_i$ to the users requesting it}
\EndFor
\EndProcedure
\end{algorithmic}
\end{algorithm}

The main idea behind the CODED DELIVERY procedure is to deliver each user the missing bits of its request that have been cached by $i$ user(s), $\forall i \in \left[ 0:K-1 \right]$. In the first part, the bits of each request that are not in the cache of any user are directly delivered. Each transmitted content is destined for all the users in a distinct group, which have the same request.

In part 2, the bits of each request that have been cached by only one user are served. Note that, for any $i \in \left[1:N \right]$, each user $U_k$ in ${\cal G}_i$, $k \in \left[ S_{i-1}+1:S_i \right]$, demands $W_i$ and has already cached $W_{i,\left\{ k \right\}}$. Thus, having received the bits delivered in line 7 of Algorithm \ref{DeliveryDecentralizedUniform}, $U_k$ can recover $W_{i,\left\{ l \right\}}$, $\forall l \in \left[ S_{i-1}+1:S_i \right]$, i.e., the bits of $W_i$ cached by all the other users in the same group. 

With the contents delivered in line 8 of Algorithm \ref{DeliveryDecentralizedUniform}, each user can decode the subfiles of its requested file, which have been cached by users in other groups. Consider the users in two different groups ${\cal G}_i$ and ${\cal G}_j$, for $i=1, ..., N-1$ and $j=i+1, ..., N$. All users in ${\cal G}_i$ can recover subfile $W_{j,\{S_i\}}$ after receiving $\bigcup\limits_{k = {S_{i - 1}} + 1}^{{S_i} - 1}  {{W_{j,\left\{ k \right\}}} \oplus {W_{j,\left\{ {k + 1} \right\}}}}$. Thus, they can obtain all subfiles $W_{i,\{l\}}$, $\forall l \in \left[ S_{j-1}+1:S_j \right]$, i.e., subfiles of $W_i$ having been cached by users in ${\cal G}_j$, after receiving $W_{i, \{S_{j-1}+1\}} \oplus W_{j, \{S_{i-1}+1\}}$ and $\mathop  \bigcup \limits_{k = S_{j - 1} + 1}^{S_j - 1} {{W_{i,\{k\}}} \oplus {W_{i,\{k + 1\}}}}$. Similarly, all users in ${\cal G}_j$ can recover $W_{i,\{S_j\}}$ after receiving $\bigcup\limits_{k = {S_{j - 1}} + 1}^{{S_j} - 1} { {{W_{i,\left\{ k \right\}}} \oplus {W_{i,\left\{ {k + 1} \right\}}}}}$. Hence, by receiving $W_{i, \{S_{j-1}+1\}} \oplus W_{j, \{S_{i-1}+1\}}$ and $\mathop  \bigcup \limits_{k = S_{i - 1} + 1}^{S_i - 1} {{W_{j,\{k\}}} \oplus {W_{j,\{k + 1\}}}}$, all users in ${\cal G}_j$ can recover all the subfiles $W_{j,\{l\}}$, $\forall l \in \left[ S_{i-1}+1:S_i \right]$, i.e., subfiles of $W_j$ that have been cached by users in ${\cal G}_i$. 

In the last part, the same procedure as the one proposed in \cite{MaddahAliDecentralized} is performed for the missing bits of each file that have been cached by more than one user. Consider any subset of users $\mathcal V \subset \left[ 1:K \right]$, such that $3 \le \left| \mathcal V \right| \le K$. Each user $v \in \cal V$ can recover subfile $W_{d_{v}, \mathcal V \backslash \left\{v\right\}}$ after receiving the coded message delivered through line 11 of Algorithm \ref{DeliveryDecentralizedUniform}. Hence, together with the local cache content and the contents delivered by the CODED DELIVERY procedure in Algorithm \ref{DeliveryDecentralizedUniform}, each user can recover its desired file.

\begin{figure*}[!t]
\centering
\includegraphics[scale=0.61,trim={0 576pt 0 114pt},clip]{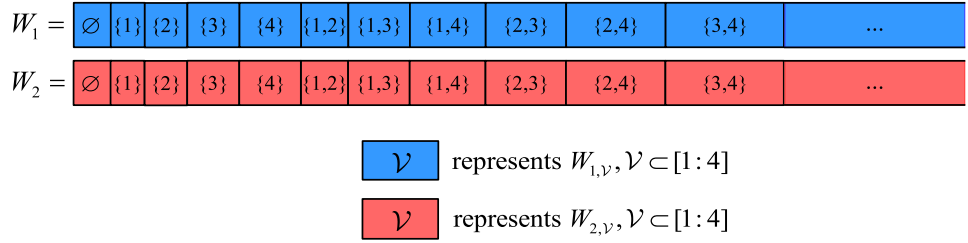}
\caption{Illustration of the subfiles, each corresponding to the bits of the file cached by a different subset of users.}
\label{DivideFileExample}
\end{figure*}

\underline{\textbf{Comparison with the state-of-the-art:}} Here we compare the delivery rate of the proposed GBD scheme with that of the scheme proposed in \cite[Algorithm 1]{MaddahAliDecentralized} for uniform cache capacities scenario. The RANDOM DELIVERY procedure in Algorithm \ref{DeliveryDecentralizedUniform} is the same as the second delivery procedure of \cite[Algorithm 1]{MaddahAliDecentralized}. Thus, we focus only on the CODED DELIVERY procedure, and compare it with the first delivery procedure in \cite[Algorithm 1]{MaddahAliDecentralized}. The two procedures differ in the first and second parts. Consider a demand combination $d_{[1:K]}$ with $N'$ different requests, i.e.,  
\begin{equation}\label{EachDemandCombinationGeneralUniformCase} d_k = i,\quad \mbox{for $i=1, ..., N,$ and $k=S_{i-1} + 1, ..., S_i$},
\end{equation} 
such that $S_i > 0$, for $i=1, ..., N'$, and $S_i = 0$, for $i=N'+1, ...,N$. In the first part of CODED DELIVERY procedure in Algorithm \ref{DeliveryDecentralizedUniform}, the bits of each $N'$ different requested files, which have not been cached by any user, are delivered. A total of $N'(1-M/N)^KF$ bits are delivered in this part. On the other hand, in the first delivery procedure in \cite[Algorithm 1]{MaddahAliDecentralized}, a total number of $K(1-M/N)^KF$ bits are delivered to serve the users with the bits which are not available in the cache of any user. From the fact that $N' \le \min\{ N,K \}$ for any demand combination, the required number of bits delivered over the shared link in the first part of the CODED DELIVERY procedure in Algorithm \ref{DeliveryDecentralizedUniform} is smaller than or equal to that of the equivalent part in \cite[Algorithm 1]{MaddahAliDecentralized}. We further note that, if $N < K$, CODED DELIVERY procedure in Algorithm \ref{DeliveryDecentralizedUniform} delivers strictly less bits than \cite[Algorithm 1]{MaddahAliDecentralized} for this part of the delivery phase.

Next, we consider the second part of the CODED DELIVERY procedure. It is shown in Appendix \ref{ProofRUc} that a total of $N'\left( K - (N'+1)/2 \right)$ coded contents, each of length $(M/N)(1-M/N)^{K-1}F$ bits, are delivered in the second part of the CODED DELIVERY procedure of the GBD scheme, leading to a delivery rate of 
\begin{equation}\label{DeliveryRateRUc}
R_{\rm{GBD}}^{\rm{U}} \buildrel \Delta \over = N' \left( K - \frac{N'+1}{2} \right)\left(\frac{M}{N}\right)\left(1-\frac{M}{N}\right)^{K-1}.
\end{equation}
On the other hand, the first delivery procedure in \cite[Algorithm 1]{MaddahAliDecentralized} sends a total of $\binom{K}{2}$ coded contents in order to serve each user with the bits that have been cached by another user, leading to a delivery rate of
\begin{equation}\label{DeliveryRateRUb}
R_{b}^{\rm{U}} \buildrel \Delta \over = \frac{K \left( K-1 \right)}{2}\left(\frac{M}{N}\right)\left(1-\frac{M}{N}\right)^{K-1}.
\end{equation}
Since, $N' \le K$, we have $R_{\rm{GBD}}^{\rm{U}} \le R_{b}^{\rm{U}}$, where the equality holds only if $N'=K$ and $N'=K-1$. Thus, for the case $N < K-1$, which results in $N' < K-1$, we have $R_{\rm{GBD}}^{\rm{U}} < R_{b}^{\rm{U}}$, in which case the second part of the CODED DELIVERY procedure of the GBD scheme requires a smaller delivery rate than that of \cite[Algorithm 1]{MaddahAliDecentralized}.         

\begin{remark}\label{RemarkComparisonUniform}
The scheme proposed in \cite{MaddahAliDecentralized} treats the users with the same demand as any other user, and it delivers the same number of bits for any demand combination; that is, for demand combination $d_{[1:K]}$, and any non-empty subset of users $\mathcal V \subset [1:K]$, it sends the coded content \begin{equation}\label{UniformSchemeStateofhteArt}
{\bigoplus} _{v \in \mathcal V} W_{d_{v},\left\{ \mathcal V \right\}\backslash \left\{v\right\}},
\end{equation}
of length $(M/N)^{\left| \mathcal V \right|-1}(1-M/N)^{K-\left| \mathcal V \right|+1}$, regardless of the redundancy among user demands. Instead, the proposed scheme treats the users with the same demand separately when delivering the bits of each file, and does not deliver redundant bits for the same demand.
\end{remark}

\subsection{Distinct Cache Capacities}\label{NonuniformScenario}
In this section, we extend the proposed GBD scheme to the scenario with distinct cache capacities. We start with an illustrative example. It is then generalized to an arbitrary network setting. A new lower bound on the delivery rate is also obtained.

\theoremstyle{definition}
\newtheorem{exmp}{Example}

\begin{exmp}\label{DecDistCacheSizes}
Consider $N = 2$ files, $W_1$ and $W_2$, and $K = 4$ users. Let the cache capacity of $U_k$ be given by ${M_k} = {\left( {1/2} \right)^{4 - k}}M_{\rm{max}}$, $\forall k \in \left[ 1:4 \right]$.


In the placement phase, $U_k$ caches a random subset of $M_{k}F/2$ bits of each file independently. Since there are $N=2$ files in the database, a total of $M_{k}F$ bits are cached by $U_k$, filling up its cache memory. File $W_i$ can be represented by
\begin{equation}\label{FileRepresentationEmaple}
{W_i} = \left( {{W_{i,\mathcal V}}:\forall \mathcal V \subset \left[ {1:4} \right]} \right), \quad \mbox{for $i=1, 2$}.
\end{equation}
An illustration of the subfiles, each cached by a different subset of users, is depicted in Fig. \ref{DivideFileExample}, and each user's cache content after the placement phase is shown in Fig. \ref{CacheContentExample}.

When $N<K$, it can be shown that the worst-case demand combination is the one when each of the $N$ users with the smallest cache capacities requests a different file. For this particular example, we have $M_1 \le \cdots \le M_4$, and accordingly, we have the worst-case demand combination when users $U_1$ and $U_2$, i.e., $N = 2$ users with the smallest cache capacities, request distinct files. Hence, we can assume the worst-case demand combination of $d_1 = d_3 = 1$ and $d_2 = d_4 = 2$.

As explained in Section \ref{UniformScenario}, the delivery phase consists of three distinct parts, where the bits delivered in part $i$, $i=1, 2, 3$, are denoted by $X(i)$, such that $X=\left( X(1), X(2), X(3) \right)$. Below, we explain the purpose of each part in detail.

\begin{figure}[!t]
\centering
\includegraphics[scale=0.4,trim={78pt 111pt 54pt 262pt},clip]{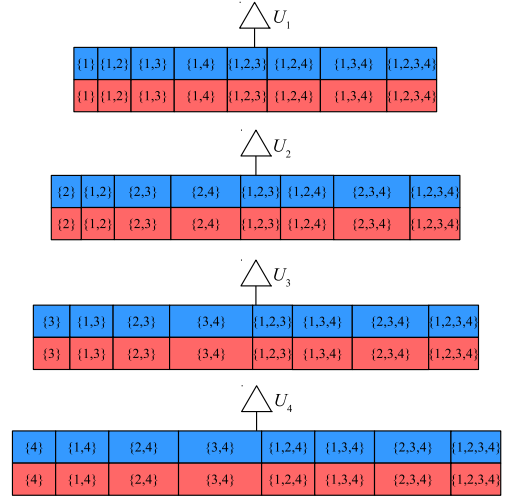}
\caption{Cache contents of users $U_{[1:4]}$ after the placement phase.} 
\label{CacheContentExample}
\end{figure}

\begin{enumerate}[label=\bfseries Part \arabic*:,align=left]

\item In the first part of the delivery phase, the bits of each requested file which have not been cached by any user are directly delivered. For the example above, the following contents are delivered: $X(1)=\left( W_{1,\emptyset}, W_{2,\emptyset} \right)$.

\item The bits of the requested files, which have been cached exclusively by a single user (other than the requester) are sent in the second part of the delivery phase. The server first delivers each user the bits of its requested file which are exclusively in the cache of one user with the same request. Then, each user receives the bits of its requested file which are only in the cache of a single user with a different request. In our example, with the following bits transmitted over the shared link, $U_k$ can recover all the bits of its request $W_{d_k}$, which have been cached exclusively by $U_l$, for $l \ne k$, $k,l \in \left[1:4\right]$: $X(2) = \left( W_{1,\left\{ 3 \right\}} \bar \oplus W_{1,\left\{ 1 \right\}}, W_{2,\left\{ 4 \right\}} \bar \oplus W_{2,\left\{ 2 \right\}}\right.$, $W_{1,\left\{ 4 \right\}} \bar \oplus W_{1,\left\{ 2 \right\}}$, $W_{2,\left\{ 3 \right\}} \bar \oplus W_{2,\left\{ 1 \right\}}$, $\left. W_{1,\left\{ 2 \right\}} \bar \oplus W_{2,\left\{ 1 \right\}} \right)$. The coded content delivered in parts 1 and 2 of the delivery phase have been illustrated in Fig. \ref{DeliveryPartsOneTwoExample}. 

\begin{figure}[!t]
\centering
\includegraphics[scale=0.4,trim={14pt 92pt 3pt 263pt},clip]{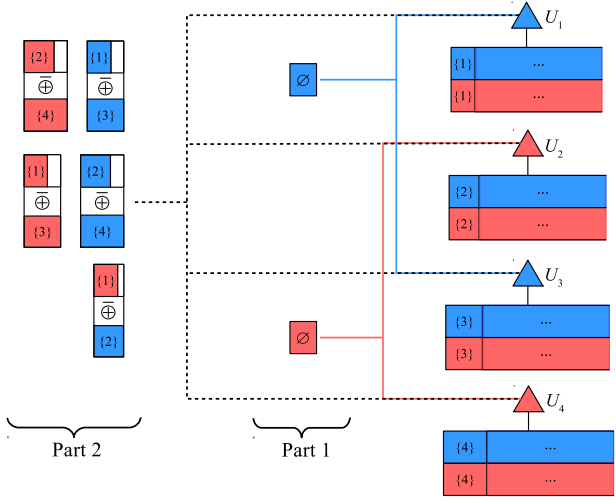}
\caption{Illustration of the coded contents delivered in parts 1 and 2 of the delivery phase for demand combination $d_1=d_3=1$ and $d_2=d_4=2$, where the cached contents are shown in Fig. \ref{CacheContentExample}.} 
\label{DeliveryPartsOneTwoExample}
\end{figure}

\item In the last part, the server delivers the users the bits of their requested files which have been cached by more than one other user. Accordingly, $U_k$, $\forall k \in \left[ 1:4 \right]$, can obtain all the bits of file $W_{d_k}$, which are in the cache of users in any set $\mathcal{S} \subset \left[ {1:4} \right]\backslash \left\{ k \right\}$, where $\left| \mathcal{S} \right| \ge 2$. For the example above, the following contents, illustrated in Fig. \ref{DeliveryPartsThreeExample}, are transmitted over the shared link: $X(3)$ $=$ $\left( W_{1,\left\{ 2,3 \right\}} \right.$ $\bar \oplus$ $W_{2,\left\{ 1,3 \right\}}$ $\bar \oplus$ $W_{1,\left\{ 1,2 \right\}}$, $W_{1,\left\{ 2,4 \right\}}$ $\bar \oplus$ $W_{2,\left\{ 1,4 \right\}}$ $\bar \oplus$ $W_{2,\left\{ 1,2 \right\}}$, $W_{1,\left\{ 3,4 \right\}}$ $\bar \oplus$ $W_{1,\left\{ 1,4 \right\}}$ $\bar \oplus$ $W_{2,\left\{ 1,3 \right\}}$, $W_{2,\left\{ 3,4 \right\}}$ $\bar \oplus$ $W_{1,\left\{ 2,4 \right\}}$ $\bar \oplus$ $W_{2,\left\{ 2,3 \right\}}$, $W_{1,\left\{ 2,3,4 \right\}}$ $\bar \oplus$ $W_{2,\left\{ 1,3,4 \right\}}$ $\bar \oplus$ $W_{1,\left\{ 1,2,4 \right\}}$ $\bar \oplus$ $\left. W_{2,\left\{ 1,2,3 \right\}} \right)$. \end{enumerate}

\begin{figure}[!t]
\centering
\includegraphics[scale=0.38,trim={5pt 111pt 5pt 264pt},clip]{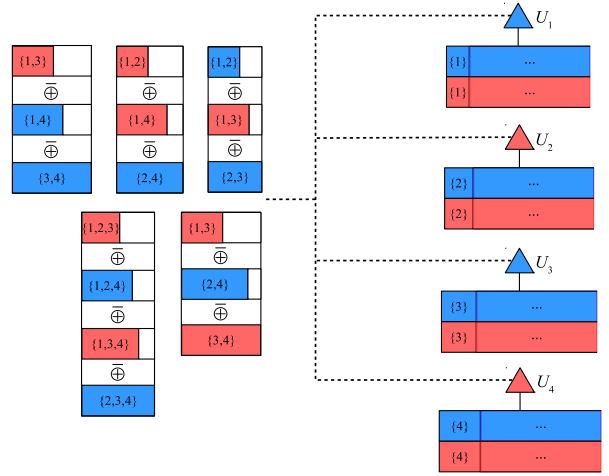}
\caption{Illustration of the coded contents delivered in part 3 of the delivery phase for demand combination $d_1=d_3=1$ and $d_2=d_4=2$, where the cached contents are shown in Fig. \ref{CacheContentExample}.} 
\label{DeliveryPartsThreeExample}
\end{figure}

After receiving these three parts, each user can decode all the missing bits of its desired file. To find the delivery rate, we first note that, for $F$ large enough, by the law of large numbers, the length of subfile $W_{k,\mathcal{V}}$, for any set $\mathcal{V} \subset \left[ {1:4} \right]$, is approximately given by
\begin{equation}\label{SizeSubfileExample} \left| {{W_{k,\mathcal{V}}}} \right| \approx \prod\limits_{i \in \mathcal{V}} {\left( {\frac{{{M_i}}}{2}} \right)} \prod\limits_{j \in \left[ {1:4} \right]\backslash \mathcal{V}} {\left( {1 - \frac{{{M_j}}}{2}} \right)} F, \quad \forall k \in \left[ {1:4} \right].
\end{equation} 
Note that, due to the $\bar \oplus$ operation, the lengths of the delivered segments, e.g., $W_{1,\left\{ 2,3 \right\}}$  $\bar \oplus$ $W_{2,\left\{ 1,3 \right\}}$ $\bar \oplus$ $W_{1,\left\{ 1,2 \right\}}$, are given by the lengths of its longest arguments, i.e., $|W_{1,\left\{ 2,3 \right\}}$ $\bar \oplus$ $W_{2,\left\{ 1,3 \right\}}$ $\bar \oplus$ $W_{1,\left\{ 1,2 \right\}}|$ $=|W_{1,\left\{ 2,3 \right\}}|$. The delivery rate is given by the total rate of all the transmitted file segments listed above. When $M_{\rm{max}} = 1$, i.e., $M_{[1:K]}  = \left( {1/8,1/4,1/2,1} \right)$, the delivery rate is $1.758$, while the scheme in \cite{WangHeterogenous} would require a delivery rate of $2.681$. The GBD scheme provides a $34.43\%$ reduction in the delivery rate compared to \cite{WangHeterogenous} in this example.  
\qed
\end{exmp}

Next, we present our caching and coded delivery scheme for the general case for arbitrary numbers of users and files, followed by the analysis of the corresponding delivery rate.

\underline{\textbf{Placement phase:}} In the placement phase, $U_k$ caches a random subset of $M_{k}F/N$ bits of each file independently, for $k=1, ..., K$. Since there are $N$ files in the database, a total of $M_{k}F$ bits are cached by $U_k$ satisfying the cache capacity constraint with equality. Since each user fills its cache independently, a bit of each file is cached exclusively by the users in set $\mathcal V \subset \left[ 1:K \right]$ (and no other user) with probability ${\prod\limits_{i \in \mathcal V} {\left( M_i/N \right)} } {\prod\limits_{j \in [1:K] \backslash \mathcal V} {\left(1 - M_j/N \right)} }$.

\underline{\textbf{Delivery phase:}}
We apply the same re-labeling of users into groups based on their requests as in Section \ref{UniformScenario}. We remind that the user demands are as follows:
\begin{equation}\label{DemandsGeneralCase} d_k = i,\quad \mbox{for $i=1, ..., N,$ and $k=S_{i-1} + 1, ..., S_i$},
\end{equation} 
where users $S_{i-1}, \ldots, S_i$ form group $\mathcal G_i$. We further order the users within a group according to their cache capacities, and assume, without loss of generality, that ${M_{{S_{i - 1}} + 1}} \le {M_{{S_{i - 1}} + 2}} \le \cdots \le {M_{{S_i}}}$, for $i= 1, \ldots ,N$.

The delivery phase of the proposed GBD scheme for distinct cache capacities is presented in Algorithm \ref{DeliveryHeterogenous}. As in Section \ref{UniformScenario}, it has two distinct delivery procedures, CODED DELIVERY and RANDOM DELIVERY; and the server chooses the one with the smaller delivery rate.

The CODED DELIVERY procedure in Algorithm \ref{DeliveryHeterogenous} follows the similar steps as the CODED DELIVERY procedure in Algorithm \ref{DeliveryDecentralizedUniform}, except that $\oplus$ is replaced with $\bar \oplus$, due to the asymmetry across the users' cache capacities, and consequently, the size of the cached subfiles by different users. We remark that the correctness of the CODED DELIVERY in Algorithm \ref{DeliveryHeterogenous} follows similarly to the correctness of the CODED DELIVERY procedure in Algorithm  \ref{DeliveryDecentralizedUniform}.

\begin{remark}\label{CodedDeliveryRemarkDistinct}
Note that $W_{i,\{ S_{j-1}+1 \}}$, $\forall i,j \in [1:N]$ such that $i \ne j$, is the smallest subfile of $W_i$ cached exclusively by one user in $\mathcal G_j$. We also note that by sending any coded content $W_{i,\{ k_1 \}} \bar \oplus W_{j,\{ k_2 \}}$, $k_1 \in \left[ S_{j-1}+1:S_j \right]$ and $k_2 \in \left[ S_{i-1}+1:S_i \right]$, instead of $W_{i,\left\{S_{j-1}+1\right\}} \bar \oplus W_{j,\left\{S_{i-1}+1\right\}}$ in $X(2,2)$, for $i=1, ..., N-1$ and $j=i+1, ..., N$, the user demands can still be satisfied. However, $W_{i,\left\{S_{j-1}+1\right\}} \bar \oplus W_{j,\left\{S_{i-1}+1\right\}}$ has the smallest length among all coded contents $W_{i,\{ k_1 \}} \bar \oplus W_{j,\{ k_2 \}}$, $\forall k_1 \in \left[ S_{j-1}+1:S_j \right]$ and $\forall k_2 \in \left[ S_{i-1}+1:S_i \right]$, which results in a smaller delivery rate.  
\end{remark}


\begin{algorithm}[t]
\caption{Coded Delivery Phase for Distinct Cache Capacities Scenario}
\label{DeliveryHeterogenous}
\begin{algorithmic}[1]
\Statex
\Procedure {Coded Delivery}{}
\State{\textbf{Part 1}: Delivering bits that are not in the cache of any user}
\For {$i = 1, 2, \ldots, N$}
\State{$X(1)=\left( W_{{d_{S_{i-1} + 1}},\emptyset}  \right)$}
\EndFor
\Statex
\State{\textbf{Part 2}: Delivering bits that are in the cache of only one user}
\State{$X(2,1)=\left( {\bigcup\limits_{i = 1}^N {\bigcup\limits_{k = {S_{i - 1}} + 1}^{{S_i} - 1} {\left( {{W_{i,\left\{ k \right\}}} \bar \oplus {W_{i,\left\{ k + 1 \right\}}}} \right)} } } \right)$}
\State{$X(2,2)=\bigcup\limits_{i = 1}^{N - 1} \bigcup\limits_{j = i + 1}^N \left( \bigcup\limits_{k = S_{j - 1} + 1}^{S_j - 1} \left( W_{i,\left\{ k \right\}} \bar \oplus W_{i,\left\{ {k + 1} \right\}} \right) ,\qquad \qquad \qquad \qquad \qquad \qquad \qquad \qquad \qquad \right.$ $ \left. \qquad \qquad \;\; \; \qquad \qquad \qquad \bigcup\limits_{k = S_{i - 1} + 1}^{S_i - 1} \left( W_{j,\left\{ k \right\}} \bar \oplus W_{j,\left\{ {k + 1} \right\}} \right), \right.$  $\left. \qquad \qquad \; \;  \qquad \qquad \qquad \left( W_{i,\left\{S_{j-1}+1\right\}} \bar \oplus W_{j,\left\{S_{i-1}+1\right\}} \right) \Bigg) \right.$} 
\Statex
\State{\textbf{Part 3}: Delivering bits that are in the cache of more than one user}
\For{$\mathcal V \subset \left[ 1:K \right]: 3 \le \left| \mathcal V \right| \le K$ }
\State $X(3) = {\overline  \bigoplus  } _{v \in \mathcal V} W_{d_{v}, \mathcal V \backslash \left\{v\right\}}$
\EndFor
\EndProcedure
\Statex
\Procedure {Random Delivery}{}
\For{$i = 1, 2, \ldots, N$}
\State {send enough random linear combinations of the bits of $W_{i}$ to enable the users demanding it to decode it}
\EndFor
\EndProcedure
\end{algorithmic}
\end{algorithm}

In the RANDOM DELIVERY procedure, as in the second delivery procedure of \cite[Algorithm 1]{MaddahAliDecentralized}, the server transmits enough random linear combinations of the bits of file $W_i$ to the users in group $\mathcal{G}_i$ such that they can all decode this file, for $i=1, \ldots, N$.

\underline{\textbf{Delivery rate analysis:}} Consider first the case $N \ge K$. It can be argued in this case that the worst-case user demands happens if each file is requested by at most one user. Hence, by re-ordering the users, for the worst-case user demands, we have $K_i = 1,$ for $1 \le i \le K$, and $K_i = 0,$ otherwise. In this case, it can be shown that the CODED DELIVERY procedure requires a lower delivery rate than the RANDOM DELIVERY procedure; hence, the server uses the former. In this case, it is possible to simplify the CODED DELIVERY procedure such that, only coded message $X(2) = \bigcup\limits_{i = 1}^{N - 1} \bigcup\limits_{j = i + 1}^N W_{i,\left\{S_{j-1}+1\right\}} \bar \oplus W_{j,\left\{S_{i-1}+1\right\}}$ is transmitted in Part 2. The corresponding common message $X=\left( X(1), X(2), X(3) \right)$ transmitted over the CODED DELIVERY procedure reduces to the delivery phase of \cite[Algorithm 2]{WangHeterogenous}. Thus, the GBD scheme achieves the same delivery rate as \cite[Algorithm 2]{WangHeterogenous} when $N \ge K$. 

Next, consider the case $N<K$. It is illustrated in Appendix \ref{Proof} that the worst-case user demands happens when $N$ users with the smallest cache capacities all request different files, i.e., they end up in different groups. The corresponding delivery rate is presented in the following theorem, the proof of which can also be found in Appendix \ref{Proof}. 

\begin{theorem}\label{TheDelRateDecDistCacheSizes}
In a decentralized content delivery network with $N$ files in the database, each of size $F$ bits, and $K$ users with cache capacities $M_{[1:K]}$ satisfying $M_1 \le M_2 \le \cdots \le M_K$, the following delivery rate-cache capacity trade-off is achievable when $N<K$:
\begin{align}\label{OurDeliveryRateHeterogenous} 
& R_{\rm{GBD}}\left(M_{[1:K]}  \right) = \min \left\{ \sum\limits_{i = 1}^K {{\prod\limits_{j = 1}^i {\left( {1 - \frac{{{M_j}}}{N}} \right)} }} \right.\nonumber\\
& \left. - \Delta {R_1}\left(M_{[1:K]}  \right) - \Delta {R_2}\left( M_{[1:K]}  \right),\sum\limits_{i = 1}^N \left( {1 - \frac{{{M_i}}}{N}} \right)  \right\},
\end{align}
where
\begin{subequations}
\label{DeltaR}
\begin{align}\label{DeltaRone}
\Delta {R_1}\left( M_{[1:K]}  \right) \buildrel \Delta \over = & \left( {K - N} \right)\prod\limits_{l = 1}^K {\left( {1 - \frac{{{M_l}}}{N}} \right)},\\
\Delta {R_2}\left( M_{[1:K]}  \right) \buildrel \Delta \over =& \left[ {\sum\limits_{k = 1}^{K - N} {(k-1) {\frac{  M_{k + N} }{{N - {M_{k + N}}}}} } } \right]\prod\limits_{l = 1}^K {\left( {1 - \frac{{{M_l}}}{N}} \right)}.
\label{DeltaRtwo}
\end{align}
\end{subequations}
\end{theorem}

\underline{\textbf{Comparison with the state-of-the-art:}} Here the proposed GBD scheme for distinct cache capacities is compared with the scheme proposed in \cite{WangHeterogenous}. We note that, although the scheme presented in \cite{WangHeterogenous} is for $N \ge K$, it can also be applied to the case $N < K$, and the delivery rate given in \cite[Theorem 2]{WangHeterogenous}, denoted here by $R_b(M_{[1:K]})$, can still be achieved. Hence, in the following, when we refer to the scheme in \cite[Algorithm 2]{WangHeterogenous}, we consider its generalization to all $N$ and $K$ values. When $N<K$, according to \cite[Theorem 2]{WangHeterogenous} and \eqref{OurDeliveryRateHeterogenous}, we have 
\begin{align}\label{DeliveryRateHeterogenousComparison} {R_b}\left( {M_{[1:K]}}  \right) - & {R_{\rm{GBD}}}\left( {M_{[1:K]}}  \right) \ge\nonumber\\
& \Delta {R_1}\left( {M_{[1:K]}}  \right) + \Delta {R_2}\left( {M_{[1:K]}}  \right) > 0.
\end{align}
The second inequality in (\ref{DeliveryRateHeterogenousComparison}) holds as long as $N < K$. Therefore, when the number of files in the database is smaller than the number of active users in the delivery phase, the GBD scheme achieves a strictly smaller delivery rate than the one presented in \cite{WangHeterogenous}.

\begin{figure}[!t]
\centering
\includegraphics[scale=0.45,trim={0pt 12pt 0pt 35pt},clip]{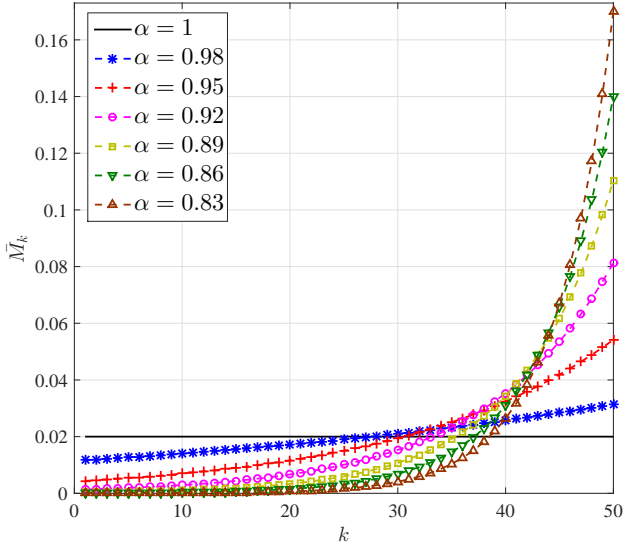}
\caption{Illustration of the normalized cache capacity distribution (normalized by $\sum\limits_{k = 1}^K {{M_k}}$) for different $\alpha$ values, in a cache network of $K=50$ users. The $x$-axis corresponds to the user index $k$.} 
\label{CacheCapacitiesDistribution}
\end{figure}

\begin{remark}\label{RemarkComparisonNonUniform}
We note that the scheme of \cite{WangHeterogenous} exploits the caching scheme of \cite{MaddahAliDecentralized} when the user cache capacities are distinct, and for any demand combination $d_{[1:K]}$, it delivers ${\overline  \bigoplus  } _{v \in \mathcal V} W_{d_{v},\left\{ \mathcal V \right\}\backslash \left\{v\right\}}$ to the users in any non-empty subset of users $\mathcal V \subset [1:K]$, regardless of the users with the same demand. Thus, for the same reason explained in Remark \ref{RemarkComparisonUniform} for uniform cache capacities, the proposed scheme in this paper outperforms the one in \cite{WangHeterogenous} for distinct cache capacities.  
\end{remark}

\subsection{Lower Bound on the Delivery Rate}\label{LowerBoundDistinctScenario}

In the next theorem, we generalize the information theoretic lower bound proposed in \cite{SenguptaCaching} to the content delivery network with distinct cache capacities. This lower bound is tighter than the classical cut-set bound.

\begin{theorem}\label{LowerBound}
In a content delivery network with $N$ files in the database, serving $K$ users with distinct cache capacities, $M_{[1:K]}$ assorted in an ascending order, the optimal delivery rate satisfies
\begin{align}\label{LowerBoundTheorem} 
& {R^*}\left(M_{[1:K]}  \right) \ge {R_{\rm{LB}}}\left( M_{[1:K]} \right) = \mathop {\max }\limits_{\scriptstyle\;\;s \in \left[ {1:K} \right],\hfill\atop
\scriptstyle l \in \left[ {1:\left\lceil {N/s} \right\rceil } \right]\hfill} \frac{1}{l}\times\nonumber\\
&\left\{ N - \frac{s}{{s + \gamma }}\sum\limits_{i = 1}^{s + \gamma} M_i - \frac{{\gamma {{\left( {N - ls} \right)}^+ }}}{{s + \gamma }} - {\left( {N - Kl} \right)}^+  \right\},
\end{align}
where $\gamma \buildrel \Delta \over = \min \left\{ {{{\left( {\left\lfloor {N/l} \right\rfloor  - s} \right)}^ + },K - s} \right\}$, $\forall s, l$. 
\end{theorem}

\begin{proof}
     The proof of the theorem can be found in Appendix \ref{ProofLowerBound}.
\end{proof}

\section{Numerical Results}\label{s:Comparison}

In this section, the proposed GBD scheme for distinct cache capacities is compared with the scheme proposed in \cite{WangHeterogenous} numerically. To highlight the gains from the proposed scheme, we also evaluate the performance of uncoded caching, in which $U_k$, $k \in \left[ 1:K \right]$, caches the first $M_k/N$ bits of each file during the placement phase; and in the delivery phase the remaining $1-M_k/N$ bits of file $W_{d_k}$ requested by $U_k$ are delivered. By a simple analysis, it can be verified that the worst-case delivery rate is given by
\begin{equation}\label{DeliveryRateUncodedCachingScheme}
{R_{uc}}\left( M_{[1:K]} \right) = \sum\limits_{i = 1}^{\min \left\{ {N,K} \right\}} {\left( {1 - \frac{{{M_i}}}{N}} \right)}, 
\end{equation}   
which is equal to the delivery rate of the RANDOM DELIVERY procedure in Algorithm \ref{DeliveryHeterogenous}.

For the numerical results, we consider an exponential cache capacity distribution among users, such that the cache capacity of $U_k$ is given by $M_k = {\alpha ^{K - k}}M_{\rm{max}}$, where $0 \le \alpha \le 1$, for $k=1, \ldots, K$, and $M_{\rm{max}}$ denotes the maximum cache capacity in the system. Thus, we have $M_{[1:K]} = \left( {{\alpha ^{K - 1}}M_{\rm{max}},{\alpha ^{K - 2}}M_{\rm{max}}, \ldots, M_{\rm{max}}} \right)$, which results in $M_1 \le M_2 \le \cdots \le M_K$, and the total cache capacity across the network is given by
\begin{equation}\label{TotalCacheCapacity}
\sum\limits_{k = 1}^K {{M_k}}  = M_{\rm{max}}{\alpha ^k}\frac{{1 - \alpha }}{{1 - {\alpha ^{K + 1}}}}.
\end{equation}
The distribution of the cache capacities normalized by the total cache capacity available in the network, denoted by $\bar{M_k} \buildrel \Delta \over = M_k/\sum\limits_{k = 1}^K {{M_k}}$, $\forall k \in \left[1:K \right]$, is demonstrated in Fig. \ref{CacheCapacitiesDistribution} for different values of $\alpha$, when $K=50$. Observe that, the smaller the value of $\alpha$, the more skewed the cache capacity distribution across the users becomes. In the special case of $\alpha=1$, we obtain the homogeneous cache capacity model studied in \cite{MaddahAliDecentralized}.

\begin{figure}[!t]
\centering
\includegraphics[scale=0.378,trim={0pt 12pt 0pt 36pt},clip]{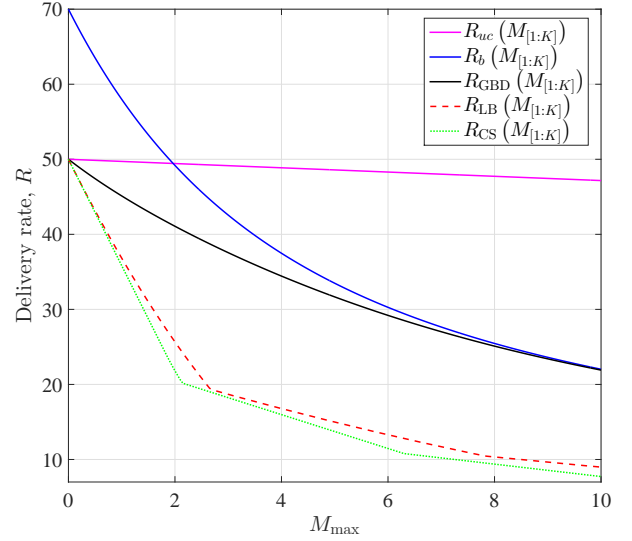}
\caption{Delivery rate versus $M_{\rm{max}}$, where the cache capacity of user $k$ is $M_k = {\alpha ^{K - k}}M_{\rm{max}}$, $k=1, \ldots, K$, with $\alpha=0.97$, $N=50$, and $K=70$.} 
\label{N50K70}
\end{figure}

In Fig. \ref{N3K3}, the delivery rate of the proposed scheme is compared with the scheme in \cite{WangHeterogenous} and uncoded caching, as well as the derived lower bound and the classical cut-set bound, when $N=K=3$ and $\alpha=0.8$. The delivery rate is plotted with respect to the largest cache capacity in the system, $M_{\rm{max}}$. As expected, the delivery rate reduces as $M_{\rm{max}}$ increases. This figure validates that the scheme proposed in Algorithm \ref{DeliveryHeterogenous} achieves the same delivery rate as in \cite{WangHeterogenous} for $N \ge K$. The GBD scheme achieves a significantly lower delivery rate compared to the uncoded scheme. It is to be noted that, the cut-set based lower bound derived in \cite{WangHeterogenous} is for the case $N \ge K$; while, by ordering the users such that $M_1 \le M_2 \le \cdots \le M_K$, it can be re-written as follows for the general case:
\begin{equation}\label{CutSetLowerBound}
{R_{\rm{CS}}}\left( M_{[1:K]}  \right) = \mathop {\max }\limits_{s \in \left[ {1:\min \left\{ {N,K} \right\}} \right]} \left\{ {s - \frac{{\sum\limits_{i = 1}^s {{M_i}} }}{{\left\lfloor {N/s} \right\rfloor }}} \right\}.
\end{equation}
The proposed lower bound is also plotted in Fig. \ref{N3K3}. Similar to the case with identical cache sizes, the proposed lower bound is tighter than the cut-set lower bound for medium cache capacities. However, there remains a gap between this improved lower bound and the achievable delivery rate for the whole range of $M_{\rm{max}}$ values.

\begin{figure}[!t]
\centering
\includegraphics[scale=0.387,trim={0pt 2pt 0pt 45pt},clip]{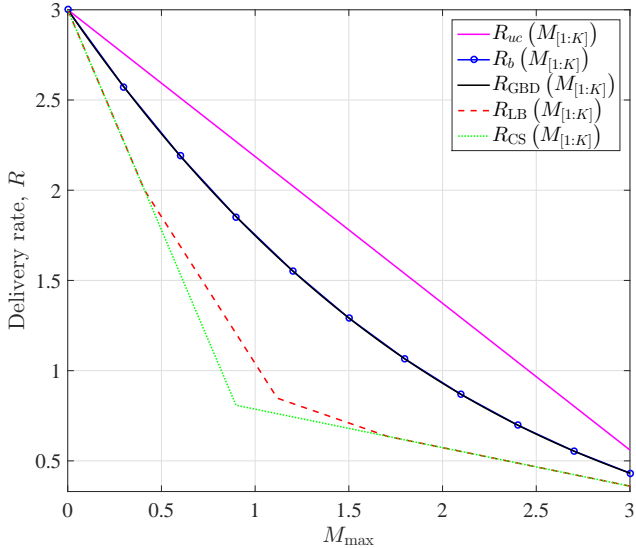}
\caption{Delivery rate versus $M_{\rm{max}}$, where the cache capacity of user $k$ is $M_k = {\alpha ^{K - k}}M_{\rm{max}}$, $k=1, \ldots, K$, when $N=K=3$ and $\alpha=0.8$.} 
\label{N3K3}
\end{figure}

In Fig. \ref{N50K70}, the delivery rate $R_{\rm{GBD}}(M_{[1:K]})$ is compared with $R_b(M_{[1:K]})$ and the uncoded scheme, $R_{uc}(M_{[1:K]})$, when $N=50$, $K=70$, and $\alpha = 0.97$. We clearly observe that the proposed scheme outperforms both schemes at all values of $M_{\rm{max}}$. The improvement is particularly significant for lower values of $M_{\rm{max}}$, and it diminishes as $M_{\rm{max}}$ increases. The proposed and the cut-set lower bounds are also included in the figure. Although the delivery rate of the proposed scheme meets the lower bounds when $M_{\rm{max}} = 0$, the gap in between quickly expands with $M_{\rm{max}}$.

In order to observe the effect of the skewness of the cache capacity distribution across users on the delivery rate, in Fig. \ref{N_75_K_90_AlphaVary}, the delivery rate is plotted as a function of $\alpha \in \left[ 0.9, 1 \right]$, for $N=30$, $K=45$, and the largest cache capacity of $M_{\rm{max}}=2$. Again, the GBD scheme achieves a lower delivery rate for the whole range of $\alpha$ values under consideration. As opposed to uncoded caching, the gain over the scheme studied in \cite{WangHeterogenous} is more pronounced for smaller values of $\alpha$, i.e., as the distribution of cache capacities becomes more skewed. We also observe that the gap to the lower bound is also smaller in this regime.

In Fig. \ref{N60KVariable}, the delivery rate is plotted with respect to the number of users, $K \in \left[ 1:100 \right]$, for $N=60$, $M_{\rm{max}}=5$, and $\alpha=0.96$. Observe that the improvement of the GBD scheme is more significant when the number of users requesting content in the delivery phase increases, whereas the gap between the GBD scheme and uncoded caching diminishes as $K$ increases.

\begin{figure}[!t]
\centering
\includegraphics[scale=0.378,trim={0pt 12pt 0pt 36pt},clip]{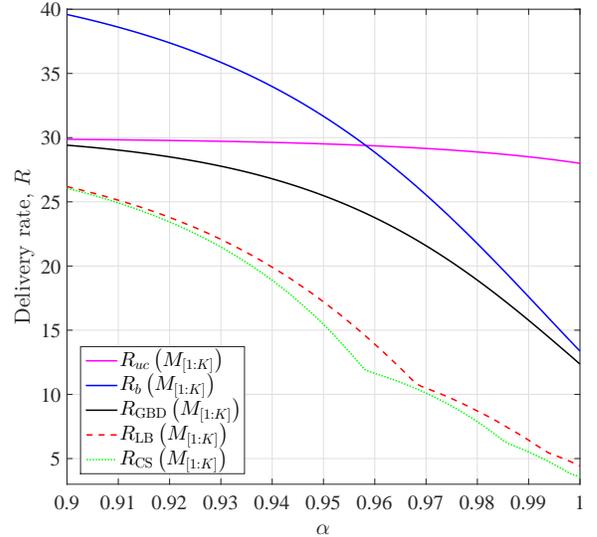}
\caption{Delivery rate versus $\alpha \in \left[ 0.9,1 \right]$, where $M_k = {\alpha ^{K - k}}M_{\rm{max}}$, for $k=1, ..., K$, and $N=30$, $K=45$, and $M_{\rm{max}}=2$.} 
\label{N_75_K_90_AlphaVary}
\end{figure}

In Fig. \ref{NVariableK50}, the delivery rate is plotted with respect to the number of files, $N \in [10:K-1]$, where the other parameters are fixed as $K=40$, $M_{\rm{max}}=4$, and $\alpha = 0.94$. We observe that, the GBD scheme requires a smaller delivery rate compared to the state-of-the-art over the whole range of $N$ values; while the improvement is more pronounced for smaller values of $N$. Observe also that, for relatively small values of $N$, the RANDOM DELIVERY procedure presented in Algorithm \ref{DeliveryHeterogenous}, which has the same performance as uncoded caching, outperforms the CODED DELIVERY procedure, i.e., $R_{\rm{RD}}\left( M_{[1:K]}  \right) < R_{\rm{CD}}\left( M_{[1:K]}  \right)$. The performance of uncoded caching gets worse with increasing $N$ in this setting.

\section{Conclusions}\label{Conc}
We have studied proactive content caching at user terminals with distinct cache capacities, and proposed a novel caching scheme in a decentralized setting that improves upon the best known delivery rate in the literature. The improvement is achieved by creating more multicasting opportunities for the delivery of bits that have not been cached by any of the users, or cached by only a single user. In particular, the proposed scheme exploits the group-based coded caching scheme we have introduced previously for centralized content caching in a system with homogeneous cache capacities in \cite{MohammadQianDenizITW}. Our numerical results show that the improvement upon the scheme proposed in \cite{WangHeterogenous} becomes more pronounced as the cache capacity distribution across users becomes more skewed, showing that the proposed scheme is more robust against variations across user capabilities. We have also derived a lower bound on the delivery rate, which has been shown numerically to be tighter than the cut-set based lower bound studied in \cite{WangHeterogenous}. The gap between the lower bound and the best achievable delivery rate remains significant, calling for more research to tighten the gap in both directions.

\begin{figure}[!t]
\centering
\includegraphics[scale=0.37,trim={0pt 12pt 0pt 36pt},clip]{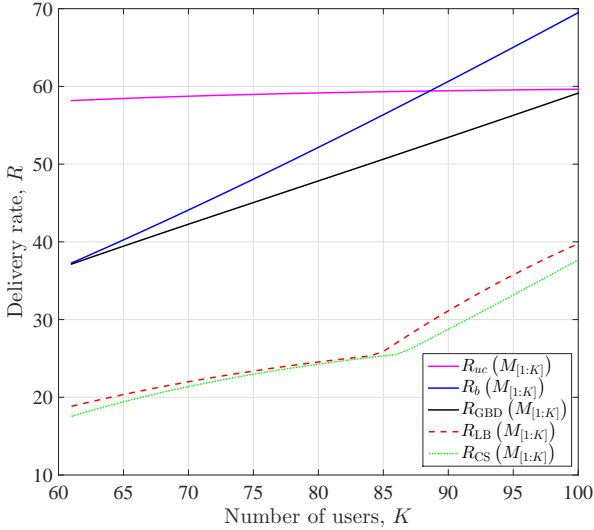}
\caption{Delivery rate versus the number of users $K \in \left[ 1:100 \right]$, where $M_k = {\alpha ^{K - k}}M_{\rm{max}}$, for $k=1, ..., K$, with $N=60$, $M_{\rm{max}}=5$, and $\alpha=0.96$.} 
\label{N60KVariable}
\end{figure}

\appendices

\section{Delivery Rate in Part 2 of the GBD Scheme}\label{ProofRUc}
If $N' \le \min \{ N,K \}$ distinct files are requested by the users, without loss of generality, we order the users so that the users in $\mathcal G_i$ request $W_i$, for $i=1, ..., N'$, i.e., $K_i=0$, for $i=N'+1, ..., N$. The coded contents delivered in line 7 of Algorithm \ref{DeliveryDecentralizedUniform} enable each user to obtain the subfiles of its requested file which are in the cache of one of the other users in the same group. Consider, for example, the first group, i.e., $i=1$ in line 7 of Algorithm \ref{DeliveryDecentralizedUniform}, which refers to the users that demand $W_1$. The XOR-ed contents $W_{1,\{k\}} \oplus W_{1,\{k+1\}}$, for $k \in \left[ {1:{K_1} - 1} \right]$, are delivered by the server. Having subfile $W_{1,\{k\}}$ cached, user $U_k$, for $k \in \left[ {1:{K_1}} \right]$, can decode all the remaining subfiles $W_{1,\{j\}}$, for $j \in \left[ {1:{K_1}} \right]\backslash \left\{ {k} \right\}$. A total of $(K_{1} - 1)$ XOR-ed contents, each of size $(M/N)(1-M/N)^{K-1}F$ bits, are delivered for the users in $\mathcal G_1$. Similarly, for the second group ($i=2$ in line 7 of Algorithm \ref{DeliveryDecentralizedUniform}), consisting of the users requesting $W_2$, the XOR-ed contents $W_{2,\{k\}} \oplus W_{2,\{k+1\}}$, for $k \in \left[ {K_{1} + 1:{K_1 + K_{2}} - 1} \right]$, are sent. With subfile $W_{2,\{k\}}$ available locally at $U_k$, for $k \in \left[ {K_{1} + 1:{K_1 + K_{2}}} \right]$, $U_k$ can obtain the missing subfiles $W_{2,\{j\}}$, for $j \in \left[ {K_{1} + 1:{K_1 + K_{2}}} \right] \backslash \left\{ {k} \right\}$. Hence, a total of $(K_2 - 1)(M/N)(1-M/N)^{K-1}F$ bits are served for the users in $\mathcal G_2$, and so on so forth. Accordingly, $(K_i - 1)(M/N)(1-M/N)^{K-1}F$ bits are delivered for group $\mathcal G_i$, $i=1, ..., N'$, and the total number of bits sent by the server in the second part of the CODED DELIVERY procedure in Algorithm \ref{DeliveryDecentralizedUniform} is 
\begin{align}\label{DeliveryRateFirstProcedure} &\left(\frac{M}{N}\right)\left(1-\frac{M}{N}\right)^{K-1}F\sum\limits_{i = 1}^{N'} {\left( {{K_i} - 1} \right)}=  \nonumber\\
& \qquad \qquad \left( K - N' \right)\left(\frac{M}{N}\right)\left(1-\frac{M}{N}\right)^{K-1}F.   
\end{align}

\begin{figure}[!t]
\centering
\includegraphics[scale=0.37,trim={0pt 12pt 0pt 36pt},clip]{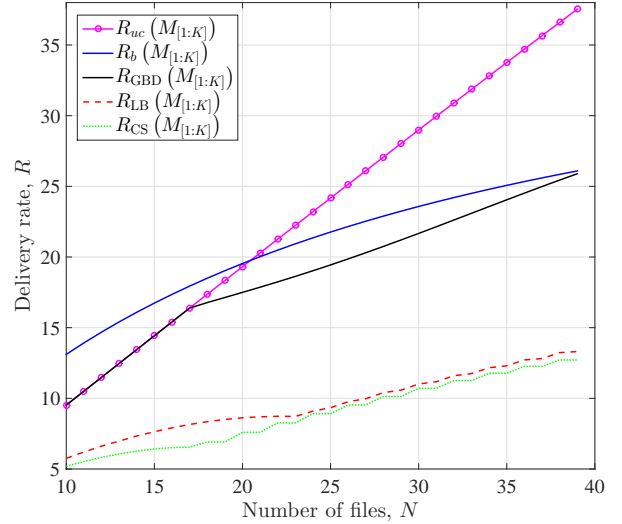}
\caption{Delivery rate versus $N \in \left[ 10:K-1 \right]$, where $M_k = {\alpha ^{K - k}}M_{\rm{max}}$, for $k=1, ..., K$, with $K=40$, $M_{\rm{max}}=4$, and $\alpha=0.94$.} 
\label{NVariableK50}
\end{figure}

After receiving the coded contents delivered in line 8 of Algorithm \ref{DeliveryDecentralizedUniform}, each user in $\mathcal G_i$, for $i \in \left[ 1:N' \right]$, can recover the missing subfiles of its request $W_i$, which are in the cache of one of the users in groups $j \in \left[ 1:N' \right] \backslash \left\{ {i} \right\}$. Consider, for example, $i=1$ and $j=2$. The XOR-ed contents $W_{1,\{k\}} \oplus W_{1,\{k+1\}}$, for $k \in \left[ {{K_1} + 1:{K_1} + {K_2} - 1} \right]$, i.e., the subfiles of $W_1$ cached by users in $\mathcal G_2$, are delivered. Next, the XOR-ed contents $W_{2,\{k\}} \oplus W_{2,\{k+1\}}$, for $k \in \left[ 1:{K_1} - 1 \right]$, i.e., the subfiles of $W_2$ cached by users in $\mathcal G_1$, are delivered. Finally, by delivering $W_{1,\{K_1 + K_2\}} \oplus W_{2,\{K_1\}}$, and having already decoded $W_{2,\{k\}}$ ($W_{1,\{k\}}$), $U_k$ in $\mathcal G_1$ ($\mathcal G_2$) can recover the missing subfiles of its request $W_1$ ($W_2$) which are in the cache of users in $\mathcal G_2$ ($\mathcal G_1$), for $k \in \left[ 1:{K_1} \right]$ (for $k \in \left[ {{K_1} + 1:{K_1} + {K_2}} \right]$). The number of coded contents delivered in line 8 is $(K_2 - 1)+ (K_1 - 1) +1$, each of length $(M/N)(1-M/N)^{K-1}F$ bits, which adds up to a total number of $(K_1 + K_2 - 1)(M/N)(1-M/N)^{K-1}F$ bits. In a similar manner, the subfiles can be exchanged between users in groups $\mathcal G_i$ and $\mathcal G_j$, for $i \in \left[ 1:N'-1 \right]$ and $j \in \left[ i + 1:N' \right]$, by delivering a total of $(K_i + K_j - 1)(M/N)(1-M/N)^{K-1}F$ bits through sending the XOR-ed contents stated in line 8 of Algorithm \ref{DeliveryDecentralizedUniform}. Hence, the total number of bits delivered in the second part of CODED DELIVERY procedure is given by   
\begin{align}\label{DeliveryRateSecondProcedure} 
& \left(\frac{M}{N}\right)\left(1-\frac{M}{N}\right)^{K-1}F\sum\limits_{i = 1}^{N' - 1} {\sum\limits_{j = i + 1}^{N'} {\left( {{K_i} + {K_j} - 1} \right)} }  =\nonumber\\
& \qquad \left( {N' - 1} \right)\left( {K - \frac{N'}{2}} \right)\left(\frac{M}{N}\right)\left(1-\frac{M}{N}\right)^{K-1}F.   
\end{align}
By summing up \eqref{DeliveryRateFirstProcedure} and \eqref{DeliveryRateSecondProcedure}, the normalized number of bits delivered in the second part of the CODED DELIVERY procedure in Algorithm \ref{DeliveryDecentralizedUniform} is given by
\begin{equation}\label{DeliveryRateEndProcedure} 
R_{\rm{GBD}}^{\rm{U}} = N' \left( K - \frac{N'+1}{2} \right)\left(\frac{M}{N}\right)\left(1-\frac{M}{N}\right)^{K-1}.
\end{equation}

\section{Proof of Theorem \ref{TheDelRateDecDistCacheSizes}}\label{Proof}
Consider first the CODED DELIVERY procedure in Algorithm \ref{DeliveryHeterogenous}. We note that, when $N<K$, the difference between the first procedure of the proposed delivery phase and the delivery phase presented in \cite[Algorithm 1]{WangHeterogenous} lies in the first two parts, i.e., delivering the missing bits of the requested files, which either have not been cached by any user, or have been cached by only a single user. Hence, having the delivery rate of the scheme in \cite[Algorithm 1]{WangHeterogenous}, the delivery rate of the CODED DELIVERY procedure in Algorithm \ref{DeliveryHeterogenous} can be determined by finding the difference in the delivery rates in these first two parts.

The delivery rate for Part 1 of the proposed CODED DELIVERY procedure, in which the bits of each request $W_{d_k}$, for $k \in [1:K]$, that have not been cached by any user are directly sent to the users requesting the file, is given by
\begin{equation}\label{OurDeliveryRateHeterogenousPartOne} {R_{{{\rm{GBD}}_1}}}\left( M_{[1:K]}  \right) = N\prod\limits_{k = 1}^K {\left( {1 - \frac{{{M_k}}}{N}} \right)}.
\end{equation}
We can see that the worst-case demand combination for this part of the CODED DELIVERY procedure is when each file is requested by at least one user, i.e., $K_i \ge 1$, $\forall i \in \left[ 1:N \right]$.

The corresponding delivery rate of \cite[Algorithm 1]{WangHeterogenous} is given by:
\begin{equation}\label{TheirDeliveryRateHeterogenousPartOne} {R_{{b_1}}}\left( M_{[1:K]}  \right) = K\prod\limits_{k = 1}^K {\left( {1 - \frac{{{M_k}}}{N}} \right)}.
\end{equation}
The difference between these two delivery rates is
\begin{align}\label{DiffDeliveryRateHeterogenousPartOne} \Delta {R_1} \left( M_{[1:K]}  \right)& \buildrel \Delta \over = {R_{{b_1}}}\left( M_{[1:K]}  \right) - {R_{{{\rm{GBD}}_1}}}\left( M_{[1:K]}  \right) \nonumber\\
& = \left( {K - N} \right)\prod\limits_{k = 1}^K {\left( {1 - \frac{{{M_k}}}{N}} \right)}.
\end{align}

In Part 2 of the delivery phase of the GBD scheme, we deal with the bits of each requested file that have been cached by only a single user $U_k$, i.e., $W_{d_j,\left\{k \right\}}$, for some $k,j \in \left[ 1:K \right]$. For any request $W_{d_j}$, the normalized number of bits that have been cached exclusively by $U_k$ will be denoted by $Q_k$. As $F \to \infty$, by the law of large numbers, $Q_k$ can be approximated as \cite{MaddahAliDecentralized}    
\begin{align}\label{DefinitionQjDeliveryRateHeterogenousPartTwo} Q_k & \approx \left( {\frac{{{M_k}}}{N}} \right)\prod\limits_{l \in \left[ {1:K} \right]\backslash \left\{ k \right\}} {\left( {1 - \frac{{{M_l}}}{N}} \right)} \nonumber\\
& = \left( {\frac{{{M_k}}}{{N - {M_k}}}} \right)\prod\limits_{l = 1}^K {\left( {1 - \frac{{{M_l}}}{N}} \right)}.
\end{align}
From \eqref{DefinitionQjDeliveryRateHeterogenousPartTwo} we can see that $Q_i \ge Q_j$, $i \ne j$, $\forall i,j \in \left[ 1:K \right]$, if and only if $M_i \ge M_j$; that is, the user with a larger cache size stores more bits of each file for $F$ sufficiently large.

Next, we evaluate the delivery rate for Part 2 of the CODED DELIVERY procedure. We start with message $X(2,1)$. For the users in $\mathcal{G}_i$, for $i=1, ..., N$, ordered in increasing cache capacities ${M_{{S_{i - 1}} + 1}} \le {M_{{S_{i - 1}} + 2}} \le \cdots \le {M_{{S_i}}}$, a total number of $\left( K_i -1 \right)$ pieces, with the normalized sizes $Q_{[S_{i-1}+2:S_{i}]}$, are delivered. Thus, the delivery rate of the common message $X(2,1)$ is given by
\begin{equation}\label{DeliveryRateHeterogenousPartTwoOne} R_{{{\rm{GBD}}_2}}^1\left( M_{[1:K]}  \right) \buildrel \Delta \over = \sum\limits_{i = 1}^N {\sum\limits_{k = {S_{i - 1}} + 2}^{{S_i}} {{Q_k}} }.
\end{equation}
In line 8 of Algorithm \ref{DeliveryHeterogenous}, $\left( K_j-1 \right)$ pieces, each of length $Q_{\left[ S_{j-1}+2:S_{j} \right]}$, and $\left( K_i-1 \right)$ pieces, each of length $Q_{\left[ S_{i-1}+2:S_{i} \right]}$ are delivered for users in $\mathcal{G}_i$ and $\mathcal{G}_j$, respectively, for $i=1, ..., N-1$ and $j=i+1, ..., N$. Hence, the rate of the common message $X(2,2)$ is given by
\begin{equation}\label{DeliveryRateHeterogenousPartTwoTwo} R_{{{\rm{GBD}}_2}}^2\left( M_{[1:K]}  \right) \buildrel \Delta \over = \sum\limits_{i = 1}^{N - 1} {\sum\limits_{j = i + 1}^N {\left( {\sum\limits_{k = {S_{j - 1}} + 2}^{{S_j}} {{Q_k}}  + \sum\limits_{k = {S_{i - 1}} + 2}^{{S_i}} {{Q_k}} } \right)} } .
\end{equation}
For each $i \in [1:N-1]$ and $j \in [i+1:N]$, the normalized length of the bits delivered with the common message $X(2,3)$ is $\max \left\{ {{Q_{{S_{j - 1}} + 1}},{Q_{{S_{i - 1}} + 1}}} \right\}$. Thus, the rate of $X(2,3)$ is found to be:
\begin{equation}\label{DeliveryRateHeterogenousPartTwoThree} R_{{{\rm{GBD}}_2}}^3\left( M_{[1:K]} \right) \buildrel \Delta \over = \sum\limits_{i = 1}^{N - 1} {\sum\limits_{j = i + 1}^N {\max \left\{ {{Q_{{S_{j - 1}} + 1}},{Q_{{S_{i - 1}} + 1}}} \right\}} }.
\end{equation}
To simplify the presentation, without loss of generality, let us assume that ${M_1} \le {M_{{S_1} + 1}} \le \cdots \le {M_{{S_{N - 1}} + 1}}$. Then \eqref{DeliveryRateHeterogenousPartTwoThree} can be rewritten as
\begin{equation}\label{DeliveryRateHeterogenousPartTwoThreeSim} R_{{{\rm{GBD}}_2}}^3\left( M_{[1:K]} \right) = \sum\limits_{i = 1}^{N - 1} {\sum\limits_{j = i + 1}^N {{Q_{{S_{j - 1}} + 1}}} }  = \sum\limits_{k = 1}^{N - 1} {k{Q_{{S_k} + 1}}}.
\end{equation}
The total delivery rate for the second part of the proposed coded delivery phase is found by summing up the rates of the three parts, i.e., 
\begin{equation}\label{DeliveryRateHeterogenousPartTwoDef} {R_{{{\rm{GBD}}_2}}}\left( M_{[1:K]}  \right) \buildrel \Delta \over = \sum_{i=1}^3 R_{{{\rm{GBD}}_2}}^i\left( M_{[1:K]}  \right).
\end{equation}
By substituting \eqref{DeliveryRateHeterogenousPartTwoOne}, \eqref{DeliveryRateHeterogenousPartTwoTwo}, and \eqref{DeliveryRateHeterogenousPartTwoThreeSim} into \eqref{DeliveryRateHeterogenousPartTwoDef}, we get 
\begin{equation}\label{DeliveryRateHeterogenousPartTwo} 
{R_{{{\rm{GBD}}_2}}}\left( M_{[1:K]} \right) = N\sum\limits_{j = 1}^N {\sum\limits_{k = {S_{j - 1}} + 2}^{{S_j}} {{Q_k}} }  + \sum\limits_{k = 1}^{N - 1} {k{Q_{{S_k} + 1}}}.
\end{equation}
Note that, in \eqref{DeliveryRateHeterogenousPartTwo}, the coefficient of $Q_{S_k+1}$ is $k$, for $k \in \left[ 0:N-1 \right]$, whereas the coefficient of all other $Q_j$s, $\forall j \in [1:K] \backslash \cal P$, where $\mathcal{P} \buildrel \Delta \over = \left\{ {1,{S_1} + 1,...,{S_{N - 1}} + 1} \right\}$, is $N$. Since $N>K$, the achievable rate for Part 2 of the CODED DELIVERY procedure in Algorithm \ref{DeliveryHeterogenous} is maximized (the worst-case user demands happens) if $Q_k \le Q_j$, for $k \in \mathcal{P}$ and $j \in \left[ 1:K \right] \backslash \mathcal{P}$; or, equivalently, if $M_k \le M_j$, for $k \in \mathcal{P}$ and $j \in \left[ 1:K \right] \backslash \mathcal{P}$. According to the definition of set $\mathcal{P}$, the above condition means that $N$ users with the smallest cache sizes, i.e., users $U_k$, $\forall k \in \mathcal{P}$, will request different files, and belong to distinct groups in the worst-case scenario. 

For simplification, without loss of generality, the users are ordered such that $M_1 \le M_2 \le \cdots \le M_K$. Then, the delivery rate of Part 2 of the CODED DELIVERY procedure is 
\begin{equation}\label{DeliveryRateHeterogenousPartTwoRelable} {R_{{{\rm{GBD}}_2}}}\left( M_{[1:K]} \right) = \sum\limits_{k = 1}^N {\left( {k - 1} \right){Q_k}}  + N\sum\limits_{k = N + 1}^K {{Q_k}}.
\end{equation}
By substituting $Q_k$ in \eqref{DefinitionQjDeliveryRateHeterogenousPartTwo}, we have 
\begin{align}\label{DeliveryRateHeterogenousPartTwoRelableMain} & R_{{{\rm{GBD}}_2}}\left( M_{[1:K]} \right) = \left[ \sum\limits_{k = 1}^N \left( {k - 1} \right)\left( {\frac{{{M_k}}}{{N - {M_k}}}} \right) +\right.\nonumber\\
& \qquad \qquad \quad \left. N\sum\limits_{k = N + 1}^K {\left( {\frac{{{M_k}}}{{N - {M_k}}}} \right)}  \right]\prod\limits_{l = 1}^K {\left( {1 - \frac{{{M_l}}}{N}} \right)}.
\end{align}

Now, we derive the delivery rate for the corresponding part in \cite[Algorithm 1]{WangHeterogenous}, i.e., when the server delivers the bits of the file requested by $U_k$, having been cached only by $U_j$, $\forall k,j \in \left[ 1:K \right]$, such that $j \ne k$. For this case, from \cite[Algorithm 1]{WangHeterogenous}, when $M_1 \le M_2 \le \cdots \le M_K$, we have
\begin{align}\label{DeliveryRateHeterogenousPartTwoTheirs} R_{{b_2}}\left( M_{[1:K]} \right) =& \left[ \sum\limits_{k = 1}^K \left( {k - 1} \right)\left( {\frac{{{M_k}}}{{N - {M_k}}}} \right)   \right]\prod\limits_{l = 1}^K {\left( {1 - \frac{{{M_l}}}{N}} \right)}.
\end{align}
Hence, the difference between the delivery rates for the second part of the proposed coded delivery phase and its counterpart in \cite[Algorithm 1]{WangHeterogenous} is given by
\begin{align}\label{DiffDeliveryRateHeterogenousPartTwo} &\Delta {R_2}\left( M_{[1:K]} \right) \buildrel \Delta \over = R_{{b_2}}\left( M_{[1:K]} \right) - R_{{{\rm{GBD}}_2}}\left( M_{[1:K]} \right) \nonumber\\
& \quad = \left[ {\sum\limits_{k = 1}^{K - N} {\left( {k - 1} \right)\left( {\frac{{{M_{k + N}}}}{{N - {M_{k + N}}}}} \right)} } \right]\prod\limits_{l = 1}^K {\left( {1 - \frac{{{M_l}}}{N}} \right)}.
\end{align}

Part 3 of the CODED DELIVERY procedure in Algorithm \ref{DeliveryHeterogenous} is the same as its counterpart in \cite[Algorithm 1]{WangHeterogenous}; so, they achieve the same delivery rate. Based on \cite[Theorem 3]{WangHeterogenous}, assuming that $M_1 \le M_2 \le \cdots \le M_K$, the delivery rate for the CODED DELIVERY procedure is 
\begin{align}\label{OurDeliveryRateHeterogenousLastProof} R_{\rm{CD}}\left( M_{[1:K]} \right) \buildrel \Delta \over = & \sum\limits_{i = 1}^K {\left[ {\prod\limits_{j = 1}^i {\left( {1 - \frac{{{M_j}}}{N}} \right)} } \right]} \nonumber\\
& - \Delta {R_1}\left( M_{[1:K]} \right) - \Delta {R_2}\left( M_{[1:K]} \right),
\end{align}
where $\Delta {R_1}\left( M_{[1:K]} \right)$ and $\Delta {R_2}\left( M_{[1:K]} \right)$ are as given in \eqref{DiffDeliveryRateHeterogenousPartOne} and \eqref{DiffDeliveryRateHeterogenousPartTwo}, respectively.

Now, consider the RANDOM DELIVERY procedure in Algorithm \ref{DeliveryHeterogenous}. Each delivered message in this procedure is directly targeted for the users in a group requesting the same file. It is assumed that the users in $\mathcal{G}_i$ are ordered to have increasing cache capacities, such that ${M_{{S_{i - 1}} + 1}} \le {M_{{S_{i - 1}} + 2}} \le \cdots \le {M_{{S_i}}}$, for $i=1, ..., N$. Since each user in $\mathcal G_i$ requires at most $\left( 1-M_{{S_{i - 1}} + 1}/N \right)F$ bits to get its requested file, a total number of $\left( 1-M_{{S_{i - 1}} + 1}/N \right)F$ bits, obtained from random linear combinations of $W_i$, are sufficient to enable the users in $\mathcal{G}_i$ to decode their request $W_i$. Hence, the delivery rate for the RANDOM DELIVERY procedure in Algorithm \ref{DeliveryHeterogenous} is 
\begin{equation}\label{DeliveryRateHeterogenousUncoded} {R_{\rm{RD}}}\left( M_{[1:K]} \right) \buildrel \Delta \over = \sum\limits_{i = 1}^N {\left( {1 - \frac{{{M_{{S_{i - 1}} + 1}}}}{N}} \right)}. 
\end{equation}
Observe that the worst-case user demand combination corresponding to delivery rate ${R_{\rm{RD}}}\left( M_{[1:K]} \right)$ happens (i.e., the delivery rate ${R_{\rm{RD}}}\left( M_{[1:K]} \right)$ is maximized) when $\left\{ M_j, \forall j \in \mathcal{P} \right\}$ forms the set of $N$ smallest cache capacities, i.e., the $N$ users with the smallest cache capacities should request different files, which is consistent with the worst-case user demand combination corresponding to ${R_{\rm{CD}}}\left( M_{[1:K]} \right)$. If the users are labelled such that $M_1 \le M_2 \le \cdots \le M_K$, then we have
\begin{equation}\label{DeliveryRateHeterogenousUncodedMain} {R_{\rm{RD}}}\left( M_{[1:K]} \right) = \sum\limits_{i = 1}^N {\left( {1 - \frac{{{M_i}}}{N}} \right)}.  
\end{equation}

We emphasize here that, before starting the \textit{delivery phase}, it is assumed that each user sends its demand, $d_k$, together with its cache contents, $Z_k$, to the server. With this information, the server can perform the delivery procedure which requiers a smaller delivery rate (by comparing \eqref{OurDeliveryRateHeterogenousLastProof} and \eqref{DeliveryRateHeterogenousUncodedMain}), and the following delivery rate is achievable:
\begin{equation}\label{OurDeliveryRateHeterogenousLastProofLast} R_{\rm{GBD}}\left( M_{[1:K]} \right) \buildrel \Delta \over = \min \left\{ R_{\rm{CD}}\left( M_{[1:K]} \right), R_{\rm{RD}}\left( M_{[1:K]} \right) \right\},
\end{equation}
which completes the proof of Theorem \ref{TheDelRateDecDistCacheSizes}.

\section{Proof of Theorem \ref{LowerBound}}\label{ProofLowerBound}
Our lower bound follows the techniques used in \cite{SenguptaCaching} to derive a lower bound for the setting with uniform cache capacities. For $s \in \left[ 1:K \right]$, it is assumed that the demands of the first $s$ users are $\left( d_{[1:s]} \right)=\left( 1, ..., s \right)$, and the remaining $\left( K-s \right)$ users have arbitrary demands $d_k \in \left[ 1:N \right]$, $\forall k \in \left[ s+1:K \right]$. The server delivers $X_1 = \psi \left(W_{[1:N]}, 1, ..., s, d_{[s+1:K]} \right)$ to serve this demand combination. Now, consider the user demands $\left( d_{[1:s]} \right)=\left( s+1, ..., 2s \right)$, and $d_k \in \left[ 1:N \right]$, $\forall k \in \left[ s+1:K \right]$, and the message $X_2 = \psi \left(W_{[1:N]}, s+1, ..., 2s, d_{[s+1:K]} \right)$ delivered by the server to satisfy this demand combination. Consequently, considering the common messages $X_{\left[1:\left\lceil {N/s} \right\rceil\right]}$ along with the cache contents $Z_{[1:s]}$, the whole database $\left\{ W_{[1:N]} \right\}$ can be recovered. We have
\begin{subequations}
\label{ProofTheorem2Procedure1}
\begin{align}\label{ProofTheorem2Procedure11}
& NF  \le H\left( {{Z_{\left[ {1:s} \right]}},{X_{\left[ {1:\left\lceil {N/s} \right\rceil } \right]}}} \right) \\
& = H\left( {{Z_{\left[ {1:s} \right]}}} \right) + H\left( {{X_{\left[ {1:\left\lceil {N/s} \right\rceil } \right]}}\left| {{Z_{\left[ {1:s} \right]}}} \right.} \right)\\
&\le \sum\limits_{i = 1}^s {{M_i}}F  + H\left( {{X_{\left[ {1:\left\lceil {N/s} \right\rceil } \right]}}\left| {{Z_{\left[ {1:s} \right]}}} \right.} \right) \label{ProofTheorem2Procedure12}\\
& = \sum\limits_{i = 1}^s {{M_i}}F  + H\left( {{X_{\left[ {1:l} \right]}}\left| {{Z_{\left[ {1:s} \right]}}} \right.} \right) \nonumber\\
& + H\left( {{X_{\left[ {l + 1:\left\lceil {N/s} \right\rceil } \right]}}\left| {{Z_{\left[ {1:s} \right]}},{X_{\left[ {1:l} \right]}}} \right.} \right)\label{ProofTheorem2Procedure13}\\
& \le \sum\limits_{i = 1}^s {{M_i}}F  + lF{R^*}\left( M_{[1:K]}  \right) \nonumber\\
& + H\left( {{X_{\left[ {l + 1:\left\lceil {N/s} \right\rceil } \right]}}\left| {{Z_{\left[ {1:s} \right]}},{X_{\left[ {1:l} \right]}}} \right.} \right)\label{ProofTheorem2Procedure14}\\
& = \sum\limits_{i = 1}^s {{M_i}}F  + lF{R^*}\left( M_{[1:K]} \right) \nonumber\\
&+ H\left( {{X_{\left[ {l + 1:\left\lceil {N/s} \right\rceil } \right]}}\left| {{Z_{\left[ {1:s} \right]}},{X_{\left[ {1:l} \right]}},{W_{\left[ {1:ls} \right]}}} \right.} \right) \nonumber\\
&+ I \left( {X_{\left[ {l + 1:\left\lceil {N/s} \right\rceil } \right]} ; W_{\left[ {1:ls} \right]} \left| {{Z_{\left[ {1:s} \right]}},{X_{\left[ {1:l} \right]}}} \right.}   \right)\label{ProofTheorem2Procedure15}
\end{align}

\begin{align}
& \le \sum\limits_{i = 1}^s {{M_i}}F  + lF{R^*}\left( M_{[1:K]} \right) \nonumber\\
& + H\left( {{X_{\left[ {l + 1:\left\lceil {N/s} \right\rceil } \right]}}\left| {{Z_{\left[ {1:s} \right]}},{X_{\left[ {1:l} \right]}},{W_{\left[ {1:ls} \right]}}} \right.} \right) \nonumber\\
& + H\left( {{W_{\left[ {1:ls} \right]}}\left| {{Z_{\left[ {1:s} \right]}},{X_{\left[ {1:l} \right]}}} \right.} \right)\label{ProofTheorem2Procedure16}\\
& \le \sum\limits_{i = 1}^s {{M_i}}F  + lF{R^*}\left( M_{[1:K]} \right)+ \varepsilon l s F + 1 \nonumber\\
&+ H\left( {{X_{\left[ {l + 1:\left\lceil {N/s} \right\rceil } \right]}}\left| {{Z_{\left[ {1:s} \right]}},{X_{\left[ {1:l} \right]}},{W_{\left[ {1:ls} \right]}}} \right.} \right) \label{ProofTheorem2Procedure17}\\
& \le \sum\limits_{i = 1}^s {{M_i}}F  + lF{R^*}\left( M_{[1:K]} \right) + \varepsilon l s F + 1\nonumber\\
& + H\left( {{X_{\left[ {l + 1:\left\lceil {N/s} \right\rceil } \right]}},{Z_{\left[ {s + 1:s + \gamma } \right]}}\left| {{Z_{\left[ {1:s} \right]}},{X_{\left[ {1:l} \right]}},{W_{\left[ {1:ls} \right]}}} \right.} \right) \label{ProofTheorem2Procedure18}\\
& = \sum\limits_{i = 1}^s {{M_i}}F  + lF{R^*}\left( M_{[1:K]} \right) + \varepsilon l s F + 1 \nonumber\\
& + H\left( {{Z_{\left[ {s + 1:s + \gamma } \right]}}\left| {{Z_{\left[ {1:s} \right]}},{X_{\left[ {1:l} \right]}},{W_{\left[ {1:ls} \right]}}} \right.} \right) \nonumber\\
& + H\left( {{X_{\left[ {l + 1:\left\lceil {N/s} \right\rceil } \right]}}\left| {{Z_{\left[ {1:s + \gamma } \right]}},{X_{\left[ {1:l} \right]}},{W_{\left[ {1:ls} \right]}}} \right.} \right)\label{ProofTheorem2Procedure19}\\
& \le \sum\limits_{i = 1}^s {{M_i}}F  + lF{R^*}\left( M_{[1:K]}  \right) + \varepsilon l s F + 1\nonumber\\
& + H\left( {{Z_{\left[ {s + 1:s + \gamma } \right]}}\left| {{Z_{\left[ {1:s} \right]}},{W_{\left[ {1:ls} \right]}}} \right.} \right) \nonumber\\
&+ H\left( {{X_{\left[ {l + 1:\left\lceil {N/s} \right\rceil } \right]}}\left| {{Z_{\left[ {1:s + \gamma } \right]}},{X_{\left[ {1:l} \right]}},{W_{\left[ {1:ls} \right]}}} \right.} \right)\label{ProofTheorem2Procedure110}\\
& = \sum\limits_{i = 1}^s {{M_i}}F  + lF{R^*}\left( M_{[1:K]}  \right) + \varepsilon l s F + 1 \nonumber\\
& + H\left( {{Z_{\left[ {1:s + \gamma } \right]}}\left| {{W_{\left[ {1:ls} \right]}}} \right.} \right) - H\left( {{Z_{\left[ {1:s} \right]}}\left| {{W_{\left[ {1:ls} \right]}}} \right.} \right) \nonumber\\
& + H\left( {{X_{\left[ {l + 1:\left\lceil {N/s} \right\rceil } \right]}}\left| {{Z_{\left[ {1:s + \gamma } \right]}},{X_{\left[ {1:l} \right]}},{W_{\left[ {1:ls} \right]}}} \right.} \right), \label{ProofTheorem2Procedure111}
\end{align}
\end{subequations} 
where $H(\cdot)$ denotes the entropy function, while $I(\cdot;\cdot)$ represents the mutual information; \eqref{ProofTheorem2Procedure13} follows from the chain rule of mutual information; \eqref{ProofTheorem2Procedure14} follows from bounding the entropy of $l$ common messages $X_{[1:l]}$ given the cache contents $Z_{[1:s]}$ by $lF{R^*}\left( M_{[1:K]} \right)$; \eqref{ProofTheorem2Procedure15} is due to the definition of the mutual information; \eqref{ProofTheorem2Procedure16} follows from the nonnegativity of entropy; \eqref{ProofTheorem2Procedure17} is obtained from Fano's inequality; and \eqref{ProofTheorem2Procedure18} also follows from the nonnegativity of entropy.  

In (d), $\gamma \le K-s$ cache contents, $Z_{\left[ s+1:s+\gamma \right]}$, are inserted inside the entropy. Note that from $Z_{\left[ s+1:s+\gamma \right]}$ together with messages $X_{[1:l]}$ the remaining $N-ls$ files in the database can be decoded. Since by each transmission $X_i$ along with the caches $Z_{[1:s+\gamma]}$, $\left( s+\gamma \right)$ files can be decoded, we have $s+\gamma \le \left\lceil N/l \right\rceil$ for $l$ number of transmissions. Hence, we have
\begin{equation}\label{GammaDefinition}
\gamma  = \min \left\{ {{{\left( {\left\lceil {\frac{N}{l}} \right\rceil  - s} \right)}^ + },K - s} \right\}.
\end{equation}
From the argument in \cite[Appendix A]{SenguptaCaching}, it can be verified that
\begin{equation}\label{BoundLastTerm}
H\left( {{X_{\left[ {l + 1:\left\lceil {N/s} \right\rceil } \right]}}\left| {{Z_{\left[ {1:s + \gamma } \right]}},{X_{\left[ {1:l} \right]}},{W_{\left[ {1:ls} \right]}}} \right.} \right) \le \left( {N - Kl} \right)^+ F.
\end{equation}
Based on \eqref{ProofTheorem2Procedure1} and \eqref{BoundLastTerm}, we have 
\begin{align}\label{ProofTheorem2Procedure2}
& NF \le \sum\limits_{i = 1}^s {{M_i}}F - H\left( {{Z_{\left[ {1:s} \right]}}\left| {{W_{\left[ {1:ls} \right]}}} \right.} \right) + lF{R^*}\left( M_{[1:K]} \right) \nonumber\\
&\; + H\left( {{Z_{\left[ {1:s + \gamma } \right]}}\left| {{W_{\left[ {1:ls} \right]}}} \right.} \right) + \left( {N - Kl} \right)^+ F + \varepsilon l s F + 1.
\end{align}
Accordingly, for any set $\mathcal{J} \subset \left[ {1:s + \gamma } \right]$ with $\left| \mathcal{J} \right| = s$, the following inequality can be derived by choosing a set of caches $\left\{ Z_\mathcal{J} \right\} = \left\{ \bigcup\limits_{k \in \mathcal{J}} {{Z_k}}  \right\}$ that allows decoding the files in the database along with $l$ common messages $X_{[1:l]}$: 
\begin{align}\label{ProofTheorem2Procedure3}
& NF \le \sum\limits_{i \in \mathcal{J}} {{M_i}}F  - H\left( {{Z_\mathcal{J}}\left| {{W_{\left[ {1:ls} \right]}}} \right.} \right) + lF{R^*}\left( M_{[1:K]} \right)  \nonumber\\
& \; \; + H\left( {{Z_{\left[ {1:s + \gamma } \right]}}\left| {{W_{\left[ {1:ls} \right]}}} \right.} \right) + {\left( {N - Kl} \right)}^+ F + \varepsilon l s F + 1.
\end{align}
Hence, there are a total number of $\binom{s+\gamma}{s}$ inequalities, each corresponding a different set $\mathcal{J}$. By taking average over all the inequalities, it can be evaluated that
\begin{align}\label{ProofTheorem2Procedure4}
NF \le & \frac{s}{{s + \gamma }}\sum\limits_{i = 1}^{s + \gamma } {{M_i}}F  - \sum\limits_{\scriptstyle \mathcal{J} \subset \left[ {1:s + \gamma } \right],\hfill\atop
\;\;\;\;\; \scriptstyle\left| \mathcal{J} \right| = s\hfill} {\frac{{H\left( {{Z_\mathcal{J}}\left| {{W_{\left[ {1:ls} \right]}}} \right.} \right)}}{{\binom{s+\gamma}{s}}}} \nonumber\\
&+ lF{R^*}\left( M_{[1:K]}  \right) + H\left( {{Z_{\left[ {1:s + \gamma } \right]}}\left| {{W_{\left[ {1:ls} \right]}}} \right.} \right) \nonumber\\
& + {\left( {N - Kl} \right)^+ }F + \varepsilon l s F + 1.
\end{align}
By applying Han's inequality \cite[Theorem 17.6.1]{CoverInformation}, we have 

\begin{align}\label{HenInequality}
\sum\limits_{\scriptstyle \mathcal{J} \subset \left[ {1:s + \gamma } \right],\hfill\atop
\;\;\;\;\; \scriptstyle\left| \mathcal{J} \right| = s\hfill} {\frac{{H\left( {{Z_\mathcal{J}}\left| {{W_{\left[ {1:ls} \right]}}} \right.} \right)}}{{\binom{s+\gamma}{s}}}}  \ge \frac{s}{{s + \gamma }}H\left( {{Z_{\left[ {1:s + \gamma } \right]}}\left| {{W_{\left[ {1:ls} \right]}}} \right.} \right).
\end{align}
Accordingly, the following lower bound can be derived:
\begin{align}\label{ProofTheorem2Procedure5}
NF \le & \frac{s}{{s + \gamma }}\sum\limits_{i = 1}^{s + \gamma } {{M_i}}F  + \frac{\gamma }{{s + \gamma }}H\left( {{Z_{\left[ {1:s + \gamma } \right]}}\left| {{W_{\left[ {1:ls} \right]}}} \right.} \right) \nonumber\\
& + lF{R^*}\left( M_{[1:K]}  \right) + {\left( {N - Kl} \right)^+ }F + \varepsilon l s F + 1.
\end{align}
It is shown in \cite[Appendix A]{SenguptaCaching} that
\begin{equation}\label{BoundMiddleTerm}
H\left( {{Z_{\left[ {1:s + \gamma } \right]}}\left| {{W_{\left[ {1:ls} \right]}}} \right.} \right) \le {\left( {N - ls} \right)^+ }F.
\end{equation}
From \eqref{ProofTheorem2Procedure5} and \eqref{BoundMiddleTerm}, we can obtain
\begin{align}\label{ProofTheorem2Procedure6}
N \le & \frac{s}{{s + \gamma }}\sum\limits_{i = 1}^{s + \gamma } {{M_i}}  + \frac{\gamma {\left( {N - ls} \right)^ + }}{{s + \gamma }} + l{R^*}\left( M_{[1:K]} \right) \nonumber\\
&+ {\left( {N - Kl} \right)^+ } + \varepsilon l s + \frac{1}{F}.
\end{align}
For $F$ large enough, $\varepsilon > 0$ is arbitrary close to zero. As a result, we have
\begin{align}\label{ProofTheorem2Procedure7}
& R^*\left( M_{[1:K]} \right) \ge \frac{1}{l}\times \nonumber\\
& \left( {N - \frac{s}{{s + \gamma }}\sum\limits_{i = 1}^{s + \gamma } {{M_i}}  - \frac{\gamma {{\left( {N - ls} \right)}^ + }}{{s + \gamma }} - {{\left( {N - Kl} \right)}^ + }} \right).
\end{align}
By optimizing over all parameters $s$, $l$, and $\gamma$, and re-ordering the users such that $M_1 \le M_2 \le \cdots \le M_K$ without loss of generality, we have
\begin{align}\label{ProofTheorem2Procedure8} 
& {R^*}\left( M_{[1:K]} \right) \ge {R_{\rm{LB}}}\left( M_{[1:K]} \right) =  \mathop {\max }\limits_{\scriptstyle\;\;s \in \left[ {1:K} \right],\hfill\atop
\scriptstyle l \in \left[ {1:\left\lceil {N/s} \right\rceil } \right]\hfill} \frac{1}{l}\times \nonumber\\
& \left\{ N - \frac{s}{{s + \gamma }}\sum\limits_{i = 1}^{s + \gamma} M_i - \frac{{\gamma {{\left( {N - ls} \right)}^+ }}}{{s + \gamma }} - {\left( {N - Kl} \right)}^+  \right\}.
\end{align}
Note that, the first $\left( s+\gamma \right)$ users have smaller cache capacities compared to all the other users. We can argue that the lower bound given in \eqref{ProofTheorem2Procedure8} is optimized over the set of cache capacities.

\bibliographystyle{IEEEtran}
\bibliography{Report}

\begin{thebibliography}{10}
\providecommand{\url}[1]{#1}
\csname url@samestyle\endcsname
\providecommand{\newblock}{\relax}
\providecommand{\bibinfo}[2]{#2}
\providecommand{\BIBentrySTDinterwordspacing}{\spaceskip=0pt\relax}
\providecommand{\BIBentryALTinterwordstretchfactor}{4}
\providecommand{\BIBentryALTinterwordspacing}{\spaceskip=\fontdimen2\font plus
\BIBentryALTinterwordstretchfactor\fontdimen3\font minus
  \fontdimen4\font\relax}
\providecommand{\BIBforeignlanguage}[2]{{%
\expandafter\ifx\csname l@#1\endcsname\relax
\typeout{** WARNING: IEEEtran.bst: No hyphenation pattern has been}%
\typeout{** loaded for the language `#1'. Using the pattern for}%
\typeout{** the default language instead.}%
\else
\language=\csname l@#1\endcsname
\fi
#2}}
\providecommand{\BIBdecl}{\relax}
\BIBdecl

\bibitem{DowdyCaching}
L.~W. Dowdy and D.~V. Foster, ``Comparative models of the file assignment
  problem,'' \emph{{ACM Comput. Surv.}}, vol.~14, pp. 287--313, Jun. 1982.

\bibitem{AlmerothCacing}
K.~C. Almeroth and M.~H. Ammar, ``The use of multicast delivery to provide a
  scalable and interactive video-on-demand service,'' \emph{{IEEE} J. Sel.
  Areas Commun.}, vol.~14, no.~6, pp. 1110--1122, Aug. 1996.

\bibitem{MaddahAliCentralized}
M.~A. Maddah-Ali and U.~Niesen, ``Fundamental limits of caching,'' \emph{{IEEE}
  Trans. Inform. Theory}, vol.~60, no.~5, pp. 2856--2867, May 2014.

\bibitem{ZhiChenXOR}
Z.~Chen, P.~Fan, and K.~B. Letaief, ``Fundamental limits of caching: Improved
  bounds for users with small buffers,'' \emph{{IET} Communications}, vol.~10,
  no.~17, pp. 2315--2318, Nov. 2016.

\bibitem{MohammadDenizISITA}
M.~{Mohammadi Amiri} and D.~G\"und\"uz, ``Improved delivery rate-cache capacity
  trade-off for centralized coded caching,'' in \emph{Proc. {IEEE} {ISITA}},
  Monterey, CA, Oct.-Nov. 2016, pp. 26--30.

\bibitem{MohammadDenizTCom}
------, ``Fundamental limits of coded caching: Improved delivery rate-cache
  capacity trade-off,'' \emph{{IEEE} Trans. Commun.}, vol.~65, no.~2, pp.
  806--815, Feb. 2017.

\bibitem{TianCentralizedCachingNew}
C.~Tian and J.~Chen, ``Caching and delivery via interference elimination,''
  \emph{{arXiv:1604.08600v1 [cs.IT]}}, Apr. 2016.

\bibitem{MohammadQianDenizITW}
M.~{Mohammadi Amiri}, Q.~Yang, and D.~G\"und\"uz, ``Coded caching for a large
  number of users,'' in \emph{Proc. IEEE ITW}, Cambridge, UK, Sep. 2016, pp.
  171--175.

\bibitem{Gomez_16}
J.~Gomez-Vilardebo, ``Fundamental limits of caching: Improved bounds with coded
  prefetching,'' \emph{arXiv:1612.09071v4}, May 2017.

\bibitem{MaddahAliDecentralized}
M.~A. Maddah-Ali and U.~Niesen, ``Decentralized caching attains order-optimal
  memory-rate tradeoff,'' \emph{{IEEE/ACM} Trans. Netw.}, vol.~23, no.~4, pp.
  1029--1040, Apr. 2014.

\bibitem{NiesenNonuniform}
U.~Niesen and M.~A. Maddah-Ali, ``Coded caching with nonuniform demands,'' in
  \emph{Proc. {IEEE} Conf. Comput. Commun. Workshops (INFOCOM WKSHPS)},
  Toronto, ON, Apr. 2014, pp. 221--226.

\bibitem{JiArXivNonuniform}
M.~Ji, A.~M. Tulino, J.~Llorca, and G.~Caire, ``Order-optimal rate of caching
  and coded multicasting with random demands,'' \emph{{arXiv: 1502.03124v1
  [cs.IT]}}, Feb. 2015.

\bibitem{ZhangDistinctFileSizes}
J.~Zhang, X.~Lin, C.~C. Wang, and X.~Wang, ``Coded caching for files with
  distinct file sizes,'' in \emph{Proc. {IEEE} Int'l Symp. on Inform. Theory},
  Hong Kong, Jun. 2015, pp. 1686--1690.

\bibitem{PedarsaniOnlineCaching}
R.~Pedarsani, M.~A. Maddah-Ali, and U.~Niesen, ``Online coded caching,'' in
  \emph{Proc. {IEEE} Int. Conf. Commun. (ICC)}, Sydney, NSW, Jun. 2014, pp.
  1878--1883.

\bibitem{QianDenizLossyJournal}
Q.~Yang and D.~G\"und\"uz, ``Coded caching and content delivery with
  heterogeneous distortion requirements,'' \emph{{arXiv:1608.05660v1 [cs.IT]}},
  Aug. 2016.

\bibitem{Joan_ICC17}
J.~P. Roig, D.~G\"und\"uz, and F.~Tosato, ``Interference networks with caches
  at both ends,'' in \emph{Proc. IEEE Int'l Conf. on Communications (ICC)},
  Paris, France, May 2017.

\bibitem{HuangFadingChannelcodedcaching}
W.~Huang, S.~Wang, L.~Ding, F.~Yang, and W.~Zhang, ``The performance analysis
  of coded cache in wireless fading channel,'' \emph{{arXiv:1504.01452v1
  [cs.IT]}}, Apr. 2015.

\bibitem{Ibrahim_WCNC17}
A.~Ibrahim, A.~Zewail, and A.~Yener, ``Centralized coded caching with
  heterogeneous cache sizes,'' in \emph{Proc. IEEE Wireless Communications and
  Networking Conference (WCNC)}, San Francisco, CA, Mar. 2017.

\bibitem{WangHeterogenous}
S.~Wang, W.~Li, X.~Tian, and H.~Liu, ``Coded caching with heterogeneous cache
  sizes,'' \emph{{arXiv:1504.01123v3 [cs.IT]}}, Aug. 2015.

\bibitem{SenguptaCaching}
A.~Sengupta, R.~Tandon, and T.~C. Clancy, ``Improved approximation of
  storage-rate tradeoff for caching via new outer bounds,'' in \emph{Proc.
  {IEEE} Int'l Symp. on Inform. Theory}, Hong Kong, Jun. 2015, pp. 1691--1695.

\bibitem{CoverInformation}
T.~M. Cover and J.~A. Thomas, \emph{Elements of Information Theory}.\hskip 1em
  plus 0.5em minus 0.4em\relax Hoboken, NJ, USA: Wiley-Interscience, John Wilet
  and Sons. Inc., 2006.

\end{thebibliography}

\end{document}